\documentclass[11pt,letterpaper,onecolumn]{IEEEtran}
\usepackage{amsmath}
\usepackage{amssymb}
\usepackage{epsfig}
\usepackage{subfig}
\usepackage{verbatim}
\usepackage{paralist}
\usepackage{tabulary}
\usepackage{cite}
\usepackage[all]{xy}

\DeclareSymbolFont{calsymbols}{OMS}{cmsy}{m}{n} \DeclareSymbolFontAlphabet{\mathcal}{calsymbols}

\def\E{\mathop\mathbb{E}\nolimits}
\def\G{\mathop\mathbb{G}\nolimits}

\begin{document}

\newtheorem{theorem}{Theorem} \newtheorem{lemma}{Lemma}
\newtheorem{definition}{Definition} \newtheorem{corollary}{Corollary}
\newtheorem{proposition}{Proposition}
\newtheorem{example}{Example}

\IEEEoverridecommandlockouts

\title{{\huge Asymptotically Optimal Policies for Hard-deadline Scheduling over Fading Channels}\thanks{The work of J.~Lee is supported by a Motorola Partnership in Research Grant.}}

\author{\authorblockN{Juyul Lee and Nihar Jindal}\\
\authorblockA{Department of Electrical and Computer Engineering\\
University of Minnesota\\
E-mail: \{juyul, nihar\}@umn.edu}}

\maketitle

\begin{abstract}
A hard-deadline, opportunistic scheduling problem in which $B$ bits must be transmitted within $T$ time-slots over a time-varying channel is studied: the transmitter must decide how many bits to serve in each slot based on knowledge of the current channel but without knowledge of the channel in future slots, with the objective of minimizing expected transmission energy.  In order to focus on the effects of delay and fading, we assume that no other packets are scheduled simultaneously and no outage is considered. We also assume that the scheduler can transmit at capacity where the underlying noise channel is Gaussian such that the energy-bit relation is a Shannon-type exponential function. No closed form solution for the optimal policy is known for this problem, which is naturally formulated as a finite-horizon dynamic program, but three different policies are shown to be optimal in the limiting regimes where $T$ is fixed and $B$ is large, $T$ is fixed and $B$ is small, and where $B$ and $T$ are simultaneously taken to infinity.  In addition, the advantage of optimal scheduling is quantified relative to a non-opportunistic (i.e., channel-blind) equal-bit policy.
\end{abstract}

\section{Introduction}

Although the basic tenants of opportunistic communication over time-varying channels are well understood, much less is known when short-term delay constraints are imposed.  Given the increasing importance of delay constrained communication, e.g., multimedia transmission, it is critical to understand how to optimize communication performance in delay-limited settings. Thereby motivated, we consider the discrete-time causal scheduling problem of transmitting a packet of $B$ bits within a hard deadline of $T$ slots over a time-varying channel. At each time slot the scheduler determines how many bits to transmit based on the current channel state information (CSI), but without future CSI, and the number of unserved bits, with the objective of minimizing the expected total energy cost. In order to focus on the interplay between opportunistic communication and delay, it is assumed that no other packets are simultaneously transmitted, and the hard deadline must always be met.

This basic problem was formulated as a finite-horizon dynamic program in \cite{Fu_WC06}, but an analytic form for the optimal scheduling policy cannot be found for most energy-bit relationships.  Indeed, such a problem is difficult to solve because the transmitter only has \textit{causal} CSI and because a particular rate must be guaranteed over a \textit{finite} time-horizon.  In our earlier work \cite{Lee_WC09}, we studied this problem in the setting where transmission occurs at the capacity of the underlying Gaussian noise channel and proposed different suboptimal scheduling policies. 

Building upon \cite{Lee_WC09}, in this work we prove the \textit{optimality} of certain scheduling policies in different asymptotic regimes.  In particular, we show that:
\begin{itemize}
    \item When the number of bits $B$ is large, the optimal scheduling policy is a linear combination of a delay-associated term and an opportunistic-term.  The opportunistic term depends on the logarithm of the channel quality, and the weight of this term decreases as the deadline approaches.
    \item When the number of bits $B$ is small, a one-shot threshold policy where all $B$ bits are transmitted in the first slot in which the channel quality is above a specified threshold is optimal.
\item When the number of bits $B$ and the time horizon $T$ are both large, a waterfilling-like policy is optimal.
\end{itemize}
 
These results are particularly important in light of the fact that the general optimal solution appears intractable.  In addition, the different asymptotically optimal schedulers provide an understanding of how 
the conflicting objectives of opportunistic communication (i.e. transmit only when the channel is strong) and delay-limited communication are optimally balanced, and how this balance depends on the time-horizon and the packet size.

In addition to showing asymptotic optimality, we also quantify the power benefits of optimal channel- and delay-aware scheduling relative to non-opportunistic equal-bit/rate transmission. These results identify that the largest benefits are obtained for severe fading, small packet size, and large time horizon. Moreover, we analyze the behavior of the scheduling policies for large and small $B$ using results on high and low SNR analysis in \cite{Shamai_IT01} and \cite{Verdu_IT02}.

\subsection{Prior Work}

The basic scheduling problem was first proposed and formulated as a finite-horizon dynamic program (DP) in \cite{Fu_WC06}. In that work a closed-form solution for the optimal scheduler is provided for the special case where the number of transmitted bits is linear in the transmit energy/power and the channel quality is restricted to integer multiples of some constant. In \cite{Zafer_WITA07}, the formulation is extended to continuous time; closed-form descriptions of the optimal policies for some specific models are found, but these do not directly apply to the discrete-time problem considered here.  In our earlier work \cite{Lee_WC09}, we specialized \cite{Fu_WC06} to the setting where the energy-bit relationship is dictated by AWGN channel capacity and proposed several different suboptimal policies.  Two of these policies are shown to be asymptotically optimal in the present work. 

Prior work has also considered the dual problem of (expected) rate maximization over a finite time horizon, i.e., the transmitter determines how to utilize a finite energy budget over a finite number of slots with the objective of maximizing the expected rate. This problem was considered in \cite{Negi_IT02}, and a one-shot threshold policy and equal power scheduling are shown to be asymptotically optimal in the low- and high-SNR regimes, respectively.  This work was extended to a multiple-access setting in \cite{Caire_IT04}.

Because transmission scheduling corresponds to power allocation, it is also useful to put the present work in the context 
of prior work on optimal power allocation in fading channels, with and without delay constraints.  In \cite{Goldsmith_IT97} it is established that waterfilling maximizes the long-term average transmitted rate; analogously, the long-term average
power needed to achieve a particular long-term average rate is minimized by waterfilling.  At the other extreme, channel inversion is known to be the optimal policy when a constant rate is desired in every fading state \cite{Caire_IT99}.  
The current setting lies between these two extremes, because our objective is to find a power allocation policy 
(based on causal CSI) such that a particular rate (i.e. $B/T$) is guaranteed over $T$ fading slots.  The case
$T=1$ clearly corresponds to zero-outage/delay-limited capacity in \cite{Caire_IT99}, while we intuitively
expect $T \rightarrow \infty$ to correspond to the long-term average rate scenario of \cite{Goldsmith_IT97}.  The
latter correspondence is made precise in Section \ref{sec-ergodic}.

\section{Problem Setup}
\label{sec:problem}

This section summarizes the scheduling problem introduced in \cite{Lee_WC09}, which is a discrete-time delay constrained scheduling problem over a wireless fading channel as illustrated in Fig.~\ref{fig:scheduling_simple}.
\begin{figure}
    \centering
    \includegraphics[width=0.45\textwidth]{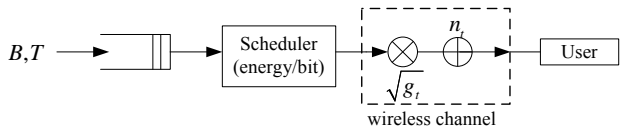}
    \caption{Point-to-point delay constrained scheduling}
    \label{fig:scheduling_simple}
\end{figure}
A packet of $B$ bits\footnote{We operate in ``nats'' instead of ``bits'' since we adopt log-base $e$ expression in the capacity formula to avoid constant factors in the analysis. We use ``bits'' and ``bit allocation'' as generic terms.} is to be transmitted within a deadline of $T$ slots. The scheduler determines the number of bits to allocate at each time slot using the fading realization/statistics to minimize the total expected transmit energy while satisfying the delay deadline constraint. We assume no other packets are to be scheduled simultaneously and that no outage is allowed.

The discrete-time slots are indexed by $t$ in descending order (i.e., starting at $t=T$ down to $t=1$), and thus $t$ represents the number of remaining slots to the deadline. The channel state (at slot $t$) is denoted by $g_t$ in power units. We assume that $g_T, g_{T-1},\cdots, g_1$ are independently and identically distributed (i.i.d.) and the probability density function (PDF) and the cumulative distribution function (CDF) are denoted by $f$ and $F$, respectively\footnote{The fading distribution must have a non-zero delay-limited capacity, i.e., $\E[1/g]<\infty$, for this problem to be feasible.}   The scheduler is assumed to have only \emph{causal} knowledge of channel states (at time $t$, $g_T,\cdots, g_t$ are known but $g_{t-1},\cdots, g_1$ are unknown). Assuming unit variance Gaussian additive noise and transmission at capacity, if  energy $E_t$ is used under channel state $g_t$, the number of transmitted bits is given by:
\begin{equation}
    b_t=\log(1+g_t E_t)
\end{equation}
By inverting this formula, the required energy $E_t$ to transmit $b_t$ bits with channel state $g_t$ is:
\begin{equation} \label{eq:Et_vs_bt}
    E_t(b_t, g_t)=\frac{e^{b_t}-1}{g_t}.
\end{equation}

The queue state is denoted by $\beta_t$, which is the number of unserved bits at the beginning of slot $t$. Thus, the number of bits to allocate at slot $t$ is determined by the queue state $\beta_t$ and the channel state $g_t$. That is, a scheduling policy is a sequence of functions, indexed by the time step, that map from the current queue and channel state to the bit allocation: $\{b_T(\beta_T,g_T), b_{T-1}(\beta_{T-1}, g_{T-1}),\cdots, b_1(\beta_1,g_1)\}$. As for terminology, the entire set $\{b_T(\cdot,\cdot), b_{T-1}(\cdot, \cdot),\cdots, b_1(\cdot,\cdot)\}$ is referred to as a \emph{policy} or a \emph{scheduler}, and each element of it is referred to as a \emph{policy function} or a \emph{scheduling function}.

\section{Optimal \& Suboptimal Schedulers}

In this section we describe the optimal scheduling policy, two suboptimal policies introduced in \cite{Lee_WC09}, and a heuristic modification of the ergodic (infinite-horizon) policy.

\subsection{The Optimal Scheduler}

The optimal scheduler for the hard-deadlined causal scheduling problem described in Section \ref{sec:problem} can be found by solving the sequential optimization:
\begin{equation} \label{eq:seq_Et}
    b_t^\text{opt}(\beta_t, g_t)= \begin{cases}
    \arg\min\limits_{0\le b_t\le \beta_t} \left\{E_t(b_t,g_t) + \E\left[\sum\limits_{s=1}^{t-1} E_s (b_s,g_s)\Bigg\vert     b_t\right]\right\}, & t=T,\ldots,2,\\
    \beta_1,&  t=1.
    \end{cases}
\end{equation}
where $\E$ denotes the expectation operator.
Equivalently, this can be formulated as a finite-horizon dynamic program (DP):
\begin{equation} \label{eq:Jt_opt}
    J_t^\text{opt}(\beta_t,g_t) =
    \begin{cases}
    \min\limits_{0\le b_t\le \beta_t} \left(\frac{e^{b_t}-1}{g_t}+\bar{J}_{t-1}^\text{opt}(\beta_t-b_t)\right), & t\ge 2\\
    \frac{e^{\beta_1}-1}{g_1}, & t=1,
    \end{cases}
\end{equation}
where
$\bar{J}_{t-1}^\text{opt}(\beta)=\E_{g}[J_{t-1}^\text{opt}(\beta,g)]$ is the cost-to-go function, i.e., the expected cost to serve $\beta$ bits in $t-1$ slots if the optimal policy is used.

At the final step ($t=1$) all $\beta_1$ remaining bits must be served because outage is not allowed. At all other steps the optimal bit allocation is determined by balancing the current energy cost $\frac{e^{b_t}-1}{g_t}$ and the expected energy expenditure in future slots $\bar{J}_{t-1}^\text{opt}(\beta_t-b_t)$.
Although the optimal scheduler can be found in closed form for $T=2$ (Section III-A in \cite{Lee_WC09}), it is not possible to do the same for $T>2$ because no close-form expression for the cost-to-go function is known for $T\ge 2$.
Nevertheless, the optimal scheduling functions can be described as \cite{Lee_WC09}:
\begin{equation} \label{eq:bt_opt_det}
    b_t^\text{opt} (\beta_t, g_t)= \begin{cases}
  0, & g_t \le \frac{1}{(\bar{J}_{t-1}^\text{opt})'(\beta_t)},\\
  \arg_b \left\{\frac{e^{b}}{g_t} = (\bar{J}_{t-1}^\text{opt})'(\beta_t-b)\right\}, &  \hspace{-10pt}\frac{1}{(\bar{J}_{t-1}^\text{opt})'(\beta_t)} < g_t < \frac{e^{\beta_t}}{(\bar{J}_{t-1}^\text{opt})'(0)}, \\
  \beta_t, & g_t \ge \frac{e^{\beta_t}}{(\bar{J}_{t-1}^\text{opt})'(0)},
 \end{cases}
 \end{equation}
where $\arg_b\{\cdot\}$ represents
the solution\footnote{Because of the convexity, the solution exists
uniquely if it exists.} of the argument equation.
The differentiability of $\bar{J}_{t-1}^\text{opt}$ can be verified by the properties of convexity and infimal convolution  (pp.~254-255 in \cite{Rockafellar_Book70}).

\vspace{10pt}
\begin{proposition} \label{prop:monotonicity_bt_opt}
    The optimal policy function $b_t^\text{opt}(\beta_t,g_t)$ has the following monotonicity properties:
    \begin{enumerate}
        \item[(a)]  For any fixed value $g_t(>0)$, $b_t^\text{opt}$ and $(\beta_t - b_t^\text{opt})$ are non-decreasing in $\beta_t$. Furthermore, there exists $\mathfrak{B}_0$ such that $b_t^\text{opt}$ and $(\beta_t - b_t^\text{opt})$ are strictly increasing in $\beta_t$ for all $\beta_t > \mathfrak{B}_0$.
        \item[(b)] For any fixed value $\beta_t(>0)$, $b_t^\text{opt}$ is non-decreasing in $g_t$.
    \end{enumerate}
\end{proposition}
\begin{proof}
    See Appendix \ref{sec:pf_monotonicity_bt_opt}.
\end{proof}
\vspace{10pt}
Intuitively, monotonicity in the queue $\beta_t$ and the channel state $g_t$ is expected because more bits should be served when there remain more unserved bits or when the channel is strong. 

\subsection{The Boundary-relaxed Scheduler}
\label{sec:boundary_relax}

The first suboptimal scheduler is derived by relaxing the boundary constraints (we no longer require $0\le b_t\le \beta_t$), while maintaining the deadline constraint $\sum_{t=1}^T b_t = B$.
The relaxed version of the original optimization \eqref{eq:Jt_opt} is given by
\begin{equation} \label{eq:Lt_opt}
    U_t(\beta_t,g_t) = \begin{cases}
        \min\limits_{b_t} \left(\frac{e^{b_t}-1}{g_t}+\bar{U}_{t-1}(\beta_t-b_t)\right), & t\ge 2,\\
        \frac{e^{\beta_1}-1}{g_1}, & t=1,
    \end{cases}
\end{equation}
where $\bar{U}_{t-1}(\beta)=\E_{g}[U_{t-1}(\beta,g)]$ and can be calculated by induction \cite{Lee_WC09}:
\begin{equation} \label{eq:barUt}
    \bar{U}_t(\beta) = t e^{\frac{\beta}{t}} \G(\nu_t,\nu_{t-1},\cdots,\nu_1) - t\nu_1,
\end{equation}
where $\G$ denotes the geometric mean operator (i.e., $\G(x_1,\cdots,x_n)=(\prod_{k=1}^n x_k)^{1/n}$) and $\nu_1,\nu_2,\cdots,\nu_t$ are the fractional moments of the fading distribution defined as:
\begin{equation} \label{eq:nu_m}
    \nu_m = \left(\E_g\left[\left(\frac{1}{g}\right)^{\frac{1}{m}}\right]\right)^m,\quad m=1,2,\cdots.
\end{equation}

Due to the simple form of the cost-to-go function $\bar{U}_t$, by substituting \eqref{eq:barUt} into \eqref{eq:Lt_opt} and solving the minimization we obtain the following closed-form description of the optimal policy for the relaxed problem \cite{Lee_WC09}:
\begin{equation} \label{eq:bt_relax_untruncated}
    b_t(\beta_t,g_t)=\frac{1}{t}\beta_t + \frac{t-1}{t}\log\left(\frac{g_t}{\eta_t^\text{relax}}\right)
\end{equation}
where $\eta_t^\text{relax}$ serves as a channel threshold given by
\begin{equation}
    \eta_t^\text{relax} = \frac{1}{\G(\nu_{t-1},\nu_{t-2},\cdots,\nu_1)}.
\end{equation}
The policy function in \eqref{eq:bt_relax_untruncated} solves the boundary-relaxed problem but does not guarantee $0\le b_t \le \beta_t$ in each slot. 

To obtain a policy for the actual unrelaxed problem, we simply truncate at 0 and $\beta_t$, and reach what we refer to as the \emph{boundary-relaxed} scheduler\footnote{This is referred to as the suboptimal II scheduler in \cite{Lee_WC09}.}:
\begin{equation} \label{eq:bt_relax}
    b_t^\text{relax}(\beta_t,g_t)=\left\langle \frac{1}{t}\beta_t + \frac{t-1}{t}\log\frac{g_t}{\eta_t^\text{relax}} \right\rangle_0^{\beta_t}
\end{equation}
where $\langle\cdot\rangle_0^{\beta_t}$ denotes truncation below $0$ and above $\beta_t$.
Notice that this policy function is optimal for $t=2$, i.e., $b_2^\text{relax}= b_2^\text{opt}$ for all $\beta_2$ and $g_2$ since $(\bar{U}_1)'=(\bar{J}_1^\text{opt})'$.

Note that this same scheduling policy can be reached using the  high-SNR approximation $\log(1 + x) \approx \log(x)$.
More specifically, if the energy-bit relationship in \eqref{eq:Et_vs_bt} is approximated by:
\begin{equation}
    E_t(b_t, g_t) = \frac{e^{b_t}-1}{g_t} \approx \frac{e^{b_t}}{g_t}.
\end{equation}
and the optimal policy is found with the same relaxation as above, the policy in \eqref{eq:bt_relax_untruncated} also
reached.

\subsection{The One-shot Scheduler}
\label{sec:one-shot}

The second scheduler is derived by modifying the boundary constraint into a stronger constraint $b_t\in\{0, \beta_t\}$ (equivalently, $b_t\in\{0, B\}$), i.e., in each slot either the entire packet is transmitted or nothing is transmitted. Then, the dynamic program is given by
\begin{equation} \label{eq:Jt_one}
    J_t^\text{one}(\beta_t,g_t) = \begin{cases}
        \min\limits_{b_t\in\{0,\beta_t\}} \left(\frac{e^{b_t}-1}{g_t}+\bar{J}_{t-1}^\text{one}(\beta_t-b_t)\right), & t\ge 2,\\
        \frac{e^{\beta_1}-1}{g_1}, & t=1,
    \end{cases}
\end{equation}
where $\bar{J}_t^\text{one}(\beta)=\E_{g}[J_t^\text{one}(\beta,g)]$. Equivalently, we can express the above DP as an optimal stopping problem \cite{Bertsekas_DP1_Book05} (this can be shown inductively with $\beta_T=B$):
\begin{equation} \label{eq:Jt_one_stopping}
    J_t^\text{one}(B,g_t) = \begin{cases}
        \min\left\{\frac{e^{B}-1}{g_t},\;\; \bar{J}_{t-1}^\text{one}(B)\right\}, & t\ge 2,\\
        \frac{e^{B}-1}{g_1}, & t=1.
    \end{cases}
\end{equation}

The optimal solution is  a \emph{sequential} threshold policy \cite{Lee_WC09}:
\begin{equation} \label{eq:bt_one}
    b_t = \begin{cases}
        B, & \text{first}\;\; t\;\;\text{such that}\;\;g_t>1/\omega_t,\\
        0,&\text{otherwise},
    \end{cases}
\end{equation}
where $1/\omega_t$ is the channel threshold in slot $t$, and is recursively computed as:
\begin{equation} \label{eq:stopping_omega_t}
\omega_t = \begin{cases}
    \E\left[\min\left(\frac{1}{g}, \omega_{t-1}\right)\right],
 & t=T,\cdots, 3, \\
 \E\left[\frac{1}{g}\right], & t=2,\\
 \infty, & t=1.
\end{cases}
\end{equation}
Notice that the thresholds depend only on the channel statistics and are independent of $B$, and that the thresholds decrease as the deadline approaches (i.e., as $t$ decreases) \cite{Lee_WC09}.

\subsection{The Delay-constrained Ergodic Scheduler}
\label{sec:constrained_erg}

The above two suboptimal policies are developed to solve the DP, formulated in \eqref{eq:Jt_opt}, by simplifying the cost-to-go function. Unlike these two policies, we now consider a policy by modifying the ergodic scheduling policy to meet the hard deadline constraint. The ergodic policy is the optimal solution to a problem of minimizing the average energy to transmit a certain \emph{average} number of bits (i.e., no hard deadline constraint). If we denote this average rate constraint as $\bar{b}$, the ergodic scheduling policy function $b(g)$, which does not depend on $t$ and determines how many bits to transmit based only upon the channel state $g$, is determined by solving:
\begin{eqnarray}
    \bar{E}^\text{erg}(\bar{b})=&\min\limits_{b(g)}& \E_g\left[\frac{e^{b(g)}-1}{g}\right] \\
    &\text{subject to}& \E_g[b(g)]\ge \bar{b}, \quad b(g)\ge 0. \nonumber
\end{eqnarray}
This optimization is readily solvable by standard waterfilling \cite{Goldsmith_Book05} and the solution is given by
\begin{equation}\label{eq:bt_erg}
    b^\text{erg}(\bar{b},g)\;\;=\;\;\left\langle \log\left(\frac{g}{\eta^\text{erg}}\right) \right\rangle_0^\infty %
    \;\;\;=\;\;\;\begin{cases}
        \log \left(\frac{g}{\eta^\text{erg}}\right), & g\ge \eta^\text{erg},\\
        0, &\text{else},
    \end{cases}
\end{equation}
where $\eta^\text{erg}$ serves as a channel threshold and is the solution to:
\begin{equation} \label{eq:b_erg_constraint}
    \E[b^\text{erg}(\bar{b},g)]=\bar{b}.
\end{equation}

When the time-horizon $T$ is large, we intuitively expect the ergodic policy to perform well in the delay-limited setting considered here. In order to meet the deadline constraint, we utilize the ergodic policy, with $\bar{b}=\frac{B}{T}+\delta$ for some $\delta>0$,\footnote{This policy is motivated by Theorem 3 of \cite{Caire_IT04}, where a modified version of the ergodic rate-maximizing policy is shown to maximize the expected transmitted rate over a finite time-horizon when the transmitter is subject to a finite energy constraint (which is the dual of the problem considered here).} at each time step with the exception that all remaining unserved bits are transmitted in the final step:
\begin{equation} \label{eq:bt_erg_delta}
    b_t^\text{constrained-erg} \left(\frac{B}{T},g_t; \delta\right) = \begin{cases}
        b^\text{erg}\left(\frac{B}{T}+\delta,g_t\right),& t=T,T-1,\cdots, 2,\\
        \beta_1,& t=1,
    \end{cases}
\end{equation}
which is referred to as the \emph{delay-constrained ergodic scheduler}.

\section{Asymptotic Optimality}

This section investigates the optimality of the suboptimal schedulers introduced in the previous section. The optimality can be analyzed in two ways: optimality in policy and optimality in the associated energy cost.
Both forms of optimality are shown for the boundary-relaxed scheduler and the one-shot scheduler, whereas energy optimality is shown for 
the delay-constrained ergodic scheduler.

\subsection{Large $B$ and Finite $T$: Asymptotic Optimality of Boundary-relaxed Scheduler}

We first prove that the boundary-relaxed scheduler converges to the optimal policy when $T$ is fixed and the number of bits $B$ is taken to infinity. When $B$ is large, we intuitively expect that the optimal policy will allocate strictly positive bits to all $T$ time slots with high probability due to the nature of the Shannon energy-bit function. Thus, we expect the boundary-relaxed scheduler to coincide with the optimal policy when the number of bits to serve is large. The following theorem makes this relationship precise:

\vspace{10pt}
\begin{theorem} \label{thm:relax_policy_conv}
    Let the PDF $f$ of $g_t$ be continuous on $[g_{\min}, g_{\max}]$ with $\text{Support}(f)=[g_{\min}, g_{\max}]$, where $g_{\min}> 0$ and $g_{\max}<\infty$.
    For every time step $t$, the boundary-relaxed policy function in \eqref{eq:bt_relax} converges to the optimal scheduling policy function uniformly on $[g_{\min}, g_{\max}]$ as the number of unserved bits $\beta$ goes to infinity: for every given $\epsilon>0$, there exists $\mathfrak{B}_0$ such that 
    \begin{equation}
        \left|b_t^\text{relax}(\beta,g_t) - b_t^\text{opt}(\beta,g_t)\right|<\epsilon, \quad \forall g_t\in [g_{\min}, g_{\max}].
    \end{equation}
    for $\beta>\mathfrak{B}_0$.
\end{theorem}
\vspace{10pt}
\begin{proof}
    See Appendix \ref{sec:pf_relax_policy_conv}.
\end{proof}
\vspace{10pt}
Figure \ref{fig:T3_b3_vs_g3}a illustrates the behaviors of $b_3^\text{relax}(\beta,g_3)$ and $b_3^\text{opt}(\beta,g_3)$ vs.~$g_3$ for different values of $\beta$ and Fig.~\ref{fig:T3_b3_vs_g3}b illustrates the behaviors in terms of $\beta$ for different values of $g_3$, when $g$ is a truncated exponential variable with a support of $[0.001, 10^6]$ (the pdf is given in \eqref{eq:trunc_pdf}). When $g_3=0.5$, for instance, it can be seen that the difference between $b_3^\text{relax}$ and $b_3^\text{opt}$ gets smaller as $\beta$ increases in both Fig.~\ref{fig:T3_b3_vs_g3}a and Fig.~\ref{fig:T3_b3_vs_g3}b. Notice also that the value of $\beta$ making the difference between $b_3^\text{relax}$ and $b_3^\text{opt}$ small varies with the value of $g_3$. As can be seen in Fig.~\ref{fig:T3_b3_vs_g3}b, larger $\beta$ is required for larger $g_3$. Additionally, we can observe from Fig.~\ref{fig:T3_b3_vs_g3}b that the slope of the plots is 1 in small $\beta$ and the slope changes to $\frac{1}{3}$ for some larger $\beta$ depending on the value of $g_3$, which is due to the policy function in \eqref{eq:bt_relax}.
\begin{figure}
    \centering
    \subfloat[$b_3^\text{relax}(\beta,g_3)$ and $b_3^\text{opt}(\beta,g_3)$ with respect to $g_3$ ]{\includegraphics[width=0.48\textwidth]{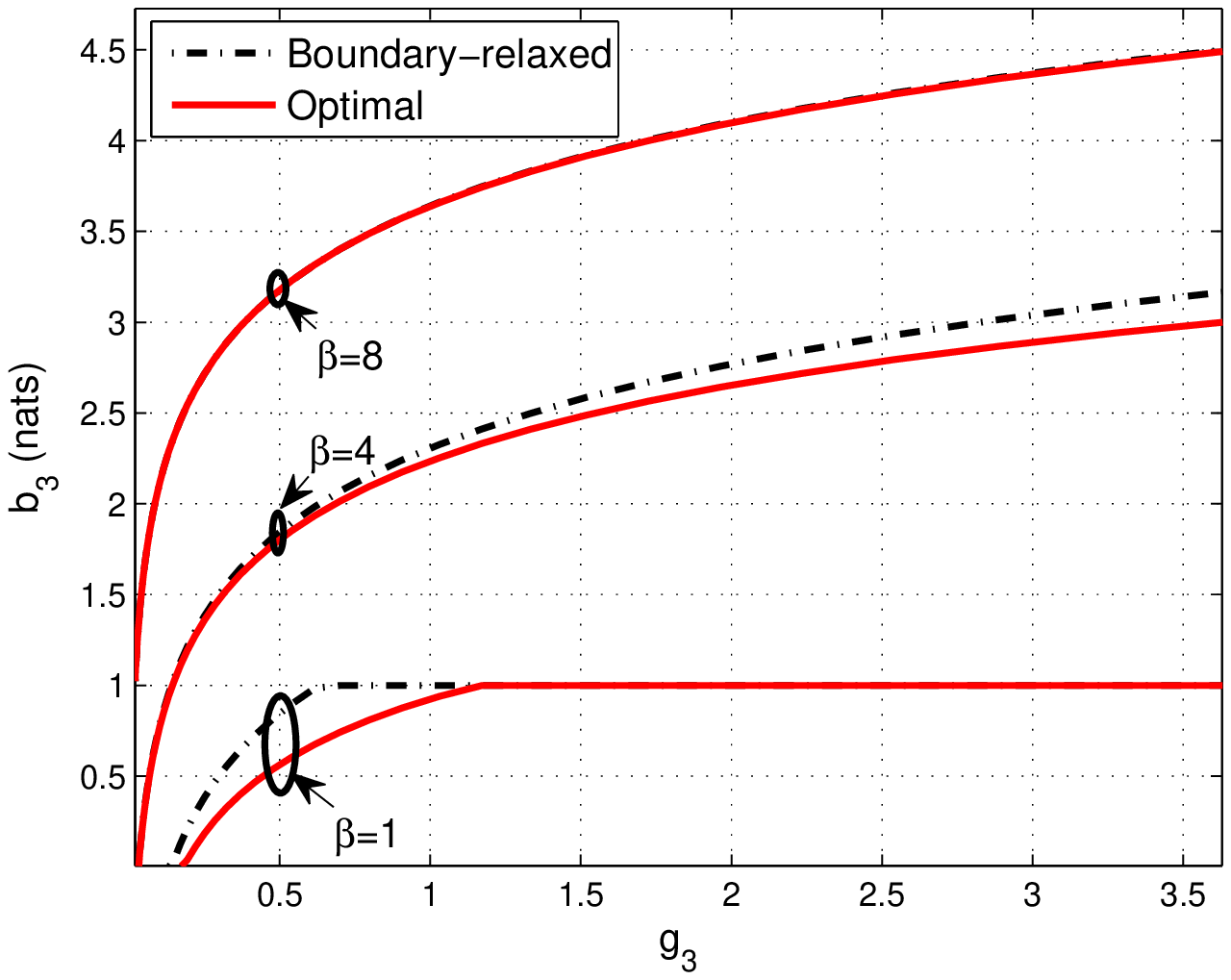}} \hfill
    \subfloat[$b_3^\text{relax}(\beta,g_3)$ and $b_3^\text{opt}(\beta,g_3)$ with respect to $\beta$]{\includegraphics[width=0.48\textwidth]{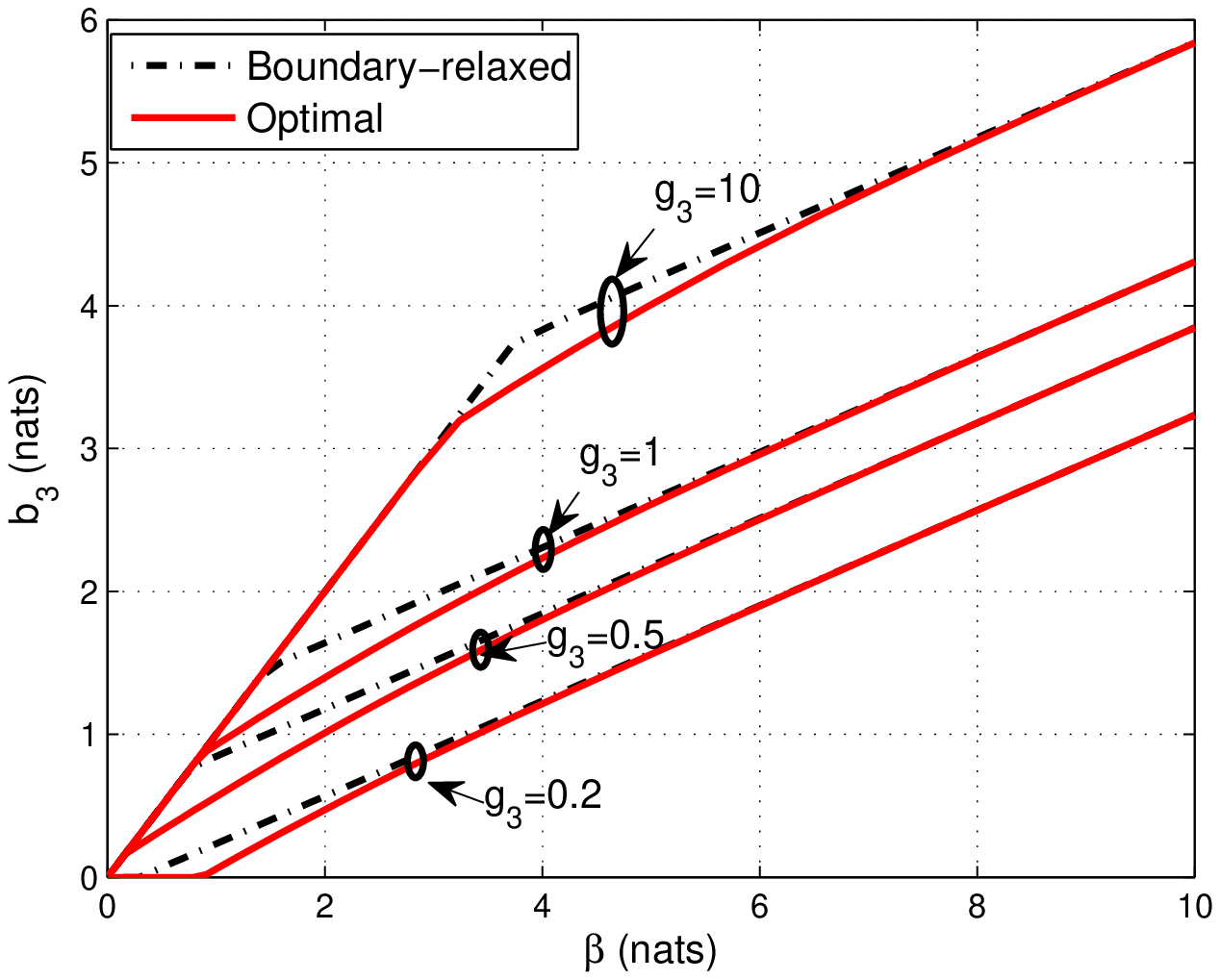}}
    \caption{The behavior of $b_3^\text{relax}$ and $b_3^\text{opt}$ when $\{g_t\}$ are truncated exponential variables with  support $[0.001, 10^6]$}
    \label{fig:T3_b3_vs_g3}
\end{figure}

We now compare the incurred energy costs of the two polices. We first define the incurred energy with the boundary-relaxed scheduler as:
\begin{equation} \label{eq:Jt_relax}
    J_t^\text{relax}(\beta_t,g_t) =
    \begin{cases}
        \frac{e^{b_t^\text{relax}}-1}{g_t} + \bar{J}_{t-1}^\text{relax}(\beta_t-b_t^\text{relax}),& t=T, T-1, \cdots, 2,\\
        \frac{e^{\beta_1}-1}{g_1}, & t=1,
    \end{cases}
\end{equation}
where $\bar{J}_{t-1}^\text{relax}(\beta)=\E_g[J_{t-1}^\text{relax}(\beta,g)]$. Notice that \eqref{eq:Jt_relax} is not an optimization but is instead a calculation based upon the definition of $b_t^\text{relax}$ in \eqref{eq:bt_relax}. 
Also notice that $\bar{J}_t^\text{relax}$ denotes the cost for the actual un-relaxed problem (the energy cost with a policy satisfying $0\le b_t\le \beta_t$ for all $t$), while the function $\bar{U}_t$ defined in Section \ref{sec:boundary_relax}  denotes the cost for the relaxed problem (the energy cost with a policy that may not satisfy $0\le b_t\le \beta_t$). 

\vspace{10pt}
\begin{theorem} \label{thm:relax_cost_conv}
    Let the PDF $f$ of $g_t$ be continuous on $[g_{\min}, g_{\max}]$ with $\text{Support}(f)=[g_{\min}, g_{\max}]$, where $g_{\min}> 0$ and $g_{\max}<\infty$.
    For any number of time slots $T$, the energy cost of the boundary-relaxed scheduler converges to the optimal energy cost as the number of bits $B$ goes to infinity:
    \begin{equation} \label{eq:lim_diff_Jt}
        \lim_{B\to\infty} \left[\bar{J}_T^\text{relax}(B)-\bar{J}_T^\text{opt}(B)\right]=0.
    \end{equation}
\end{theorem}
\vspace{10pt}
\begin{proof}
    See Appendix \ref{sec:pf_relax_cost_conv}.
\end{proof}
\vspace{10pt}
While proving Theorem \ref{thm:relax_cost_conv}, we obtain the asymptotic relations between the actual cost of the boundary-relaxed scheduler, the cost of the relaxed version, and the cost of the optimal one, i.e., $\lim_{B\to\infty} \left[\bar{J}_T^\text{relax}(B)-\bar{U}_T(B)\right]=0$ and $\lim_{B\to\infty} \left[\bar{U}_T(B)-\bar{J}_T^\text{opt}(B)\right]=0$. Since we have a closed-form expression of $\bar{U}_T(B)$ shown in \eqref{eq:barUt}, these relations help us understand the behavior of the optimal cost for large $B$, which will be discussed in Section \ref{sec:largeB}.

Although the analytic form of the optimal scheduler is not available, the above two theorems tell us that the boundary-relaxed scheduler, which has a very simple form that can be easily implemented, is asymptotically optimal when the number of bits to transmit $(B)$ is sufficiently large.
Furthermore, the scheduling function \eqref{eq:bt_relax} provides intuition on the interplay between the channel quality and the deadline.
When the deadline is far away (large $t$), the bit allocation is almost completely determined by the channel quality; on the other hand, as the deadline approaches (small $t$), the policy becomes less opportunistic.

\subsection{Small $B$ and Finite $T$: Asymptotic Optimality of One-shot Scheduler}

We now show that the one-shot scheduling policy is asymptotically optimal when $T$ is fixed and $B$ is taken to zero.  We first show convergence in terms of the policy function, and then in terms of the energy cost.

\vspace{10pt}
\begin{theorem} \label{thm:bt_one_conv_policy}
    For arbitrary time step $t$, the one-shot policy function in \eqref{eq:bt_one} converges to the optimal scheduling policy function as the number of unserved bits $\beta$ tends to zero, i.e., the optimal policy becomes a threshold policy and the threshold coincides with the threshold of the one-shot policy:
    \begin{equation}
        \lim_{\beta\to 0} \sup \{g: b_t^\text{opt}(\beta,g) =0\} = \lim_{\beta\to 0} \inf \{g : b_t^\text{opt}(\beta,g)=\beta \} =\frac{1}{\omega_t},
    \end{equation}
    where $1/\omega_t$ is the threshold of the one-shot policy as in \eqref{eq:bt_one} and \eqref{eq:stopping_omega_t}.
\end{theorem}
\vspace{5pt}
\begin{proof}
    See Appendix \ref{sec:pf_bt_one_conv_policy}.
\end{proof}

Furthermore, we claim that the costs of the two policies also converge to one another. Since the average costs for the two policies converge to zero as $B\to 0$, cost convergence is investigated by studying the ratio, rather than the absolute difference, between the two costs:
\vspace{5pt}
\begin{theorem} \label{thm:oneshot_cost_conv}
    For arbitrary delay deadline $T$, the energy cost of the one-shot scheduler converges to the optimal energy cost as the number of bits $B$ goes to zero:
    \begin{equation}
        \lim_{B\to 0} \frac{\bar{J}_T^\text{one}(B)}{\bar{J}_T^\text{opt}(B)}=1.
    \end{equation}
\end{theorem}
\vspace{5pt}
\begin{proof}
    See Appendix \ref{sec:pf_oneshot_cost_conv}.
\end{proof}
In Fig.~\ref{fig:simul_barJT_one} the additional power cost of one-shot scheduling relative to optimal scheduling  (i.e., $10\log_{10}\frac{\bar{J}_T^\text{one}(B)}{\bar{J}_T^\text{opt}(B)}$) is plotted versus the number of bits $B$ for $T=2$ and $T=3$ when $g$ is a truncated exponential variable with a support of $[0.001, 10^6]$. As can be seen, the ratio converges to 1 (0 dB) as $B$ converges to 0.
\begin{figure}
    \centering
    \includegraphics[width=0.6\textwidth]{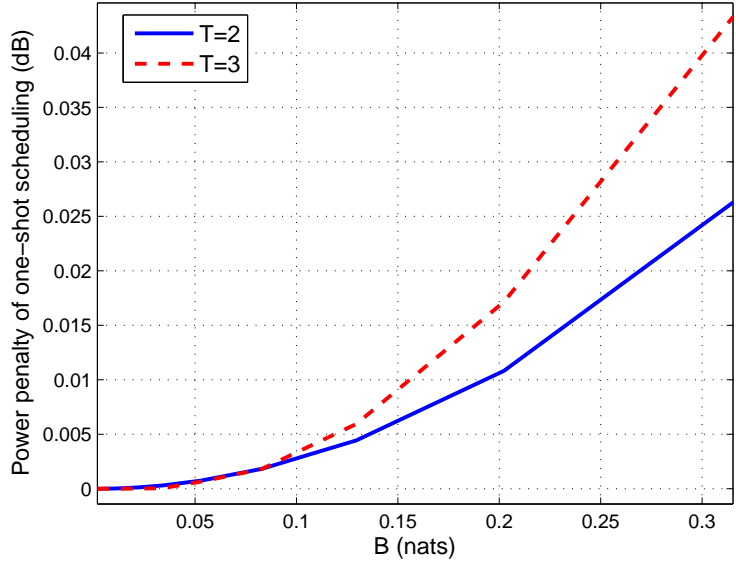}
    \caption{Additional power cost of one-shot scheduling relative to optimal scheduling as a function of $B$, when $g$ is a truncated exponential variable with support $[0.001, 10^6]$}
    \label{fig:simul_barJT_one}
\end{figure}

The optimality of one-shot scheduling can also be seen by upper and lower bounding the energy-bit function by linear functions. Using $x\le e^x-1 \le x e^B$ for $0\le x \le B$, we have:
\begin{equation}
    \frac{b_t}{g_t} \le E_t(b_t, g_t) \le \frac{b_t e^B}{g_t}.
\end{equation}

If we solve the DP using either of these bounds on the energy-bit function, the optimization in \eqref{eq:Jt_opt} becomes a linear program and thus a one-shot policy is optimal because a constrained linear program has a solution at a boundary of the constraint. Furthermore, the one-shot policy based on the upper and lower bounds converge to the one-shot policy described in Section \ref{sec:one-shot} as $B\to 0$ because the bounds themselves converge.

\subsection{Large $T$: Asymptotic Optimality of Causal Delay-constrained Ergodic Scheduler}
\label{sec-ergodic}

When $B$ and $T$ are simultaneously taken to infinity at a particular ratio (i.e., $B, T\to \infty$ with $B=\bar{b}T$ for some constant $\bar{b}>0$), we can show the energy-cost optimality of the ergodic policy in Section \ref{sec:constrained_erg}.

The average energy cost of the delay-constrained ergodic scheduler is given by
\begin{equation} \label{eq:barJT_constrained_delta}
    \bar{J}_T^\text{constrained-erg}(\bar{b}T; \delta) = \E\left[ \sum_{t=1}^T \frac{e^{b_t^\text{constrained-erg}}-1}{g_t} \right]
    = \E\left[ \sum_{t=2}^T \frac{e^{b_t^\text{erg}(\bar{b}+\delta, g_t)}-1}{g_t}\right] + \E\left[\frac{e^{\beta_1}-1}{g_1}\right],
\end{equation}
where $\beta_1$ denotes the remaining bits at the final slot and the value of $\delta$ is chosen such that
\begin{equation} \label{eq:barJT_constrained}
    \bar{J}_T^\text{constrained-erg}(\bar{b}T) = \inf_{\delta>0} \bar{J}_T^\text{constrained-erg}(\bar{b}T; \delta).
\end{equation}
\vspace{10pt}
\begin{theorem} \label{thm:J_ergdeltaopt_conv}
    For any given average rate $\bar{b}(>0)$, the per-slot energy cost of the delay-constrained ergodic policy converges to the optimal ergodic energy cost as $T$ tends to infinity:
    \begin{equation} \label{eq:perslot_JT_ergdelta_to_E_erg}
        \lim_{T\to\infty} \frac{1}{T}\bar{J}_T^\text{constrained-erg}(\bar{b}T) = \lim_{T\to\infty} \frac{1}{T} \bar{J}_T^\text{opt}(\bar{b}T)=\bar{E}^\text{erg}(\bar{b}).
    \end{equation}
\end{theorem}
\vspace{10pt}
\begin{proof}
    See Appendix \ref{sec:pf_J_ergdeltaopt_conv}.
\end{proof}
\vspace{10pt}
The effect of the hard-deadline becomes inconsequential for large $T$ because the channel realizations over the deadline horizon closely match the fading distribution. As a result, the delay-constrained ergodic scheduler performs similar to the ergodic scheduler when $T$ is large. Moreover, the delay-constrained ergodic scheduler becomes causal optimal since any causal policy cannot be better than the ergodic policy.

\subsection{Numerical Results: Policy Comparison}

In order to compare the different asymptotically optimal policies, we compare their respective energy costs for different time-horizons $(T)$. Since the analytical expression for the optimal policy is not available for $T>2$, we solve the dynamic programming \eqref{eq:Jt_opt} numerically by the discretization method \cite{Bertsekas_AC75}. In Fig.~\ref{fig:simul_JT_one_delta_relaxed_Blin} the per-slot energy consumption of
\begin{figure}
    \centering
    \subfloat[$T=5$]{\includegraphics[width=0.48\textwidth]{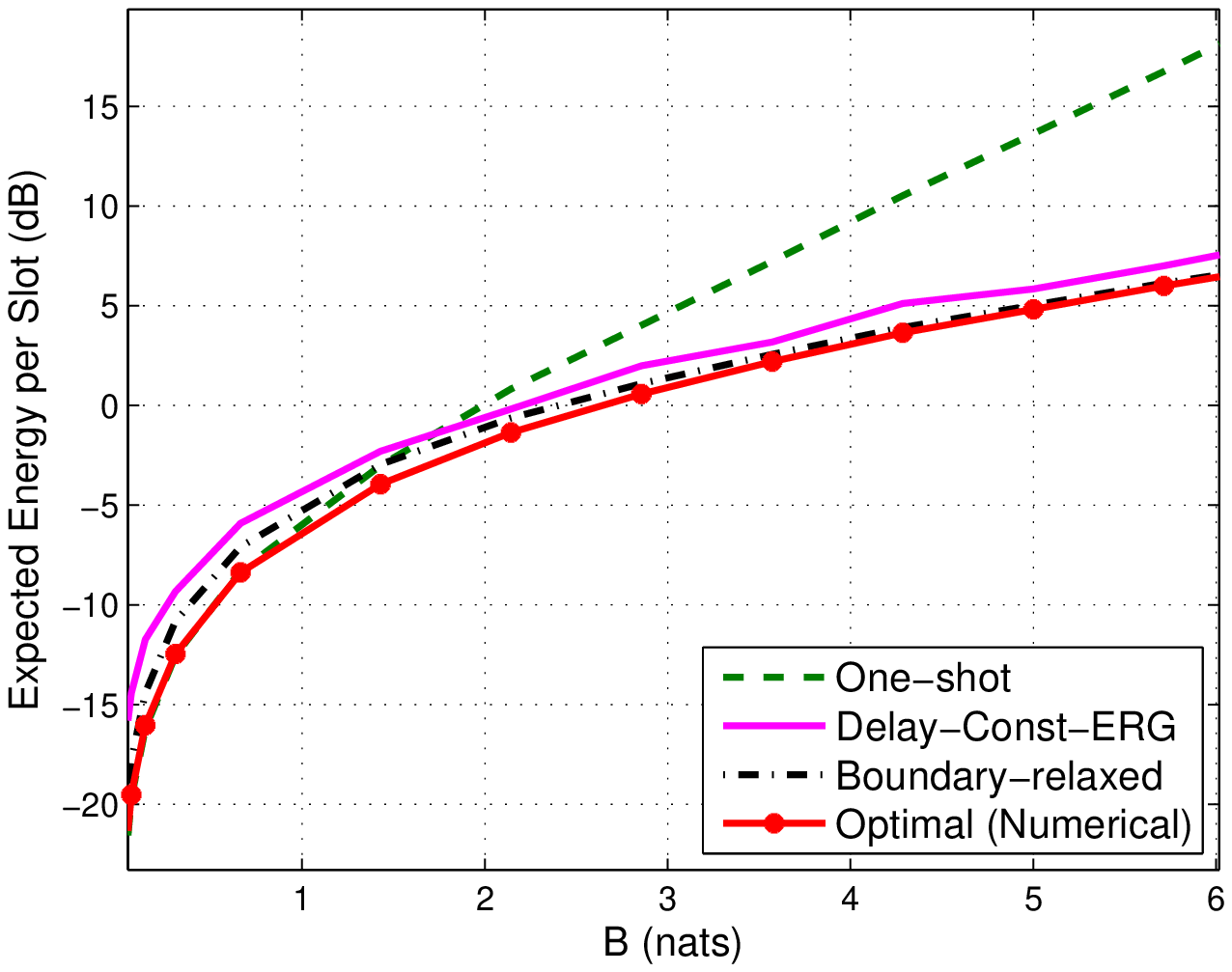}} \hfill
    \subfloat[$T=50$]{\includegraphics[width=0.48\textwidth]{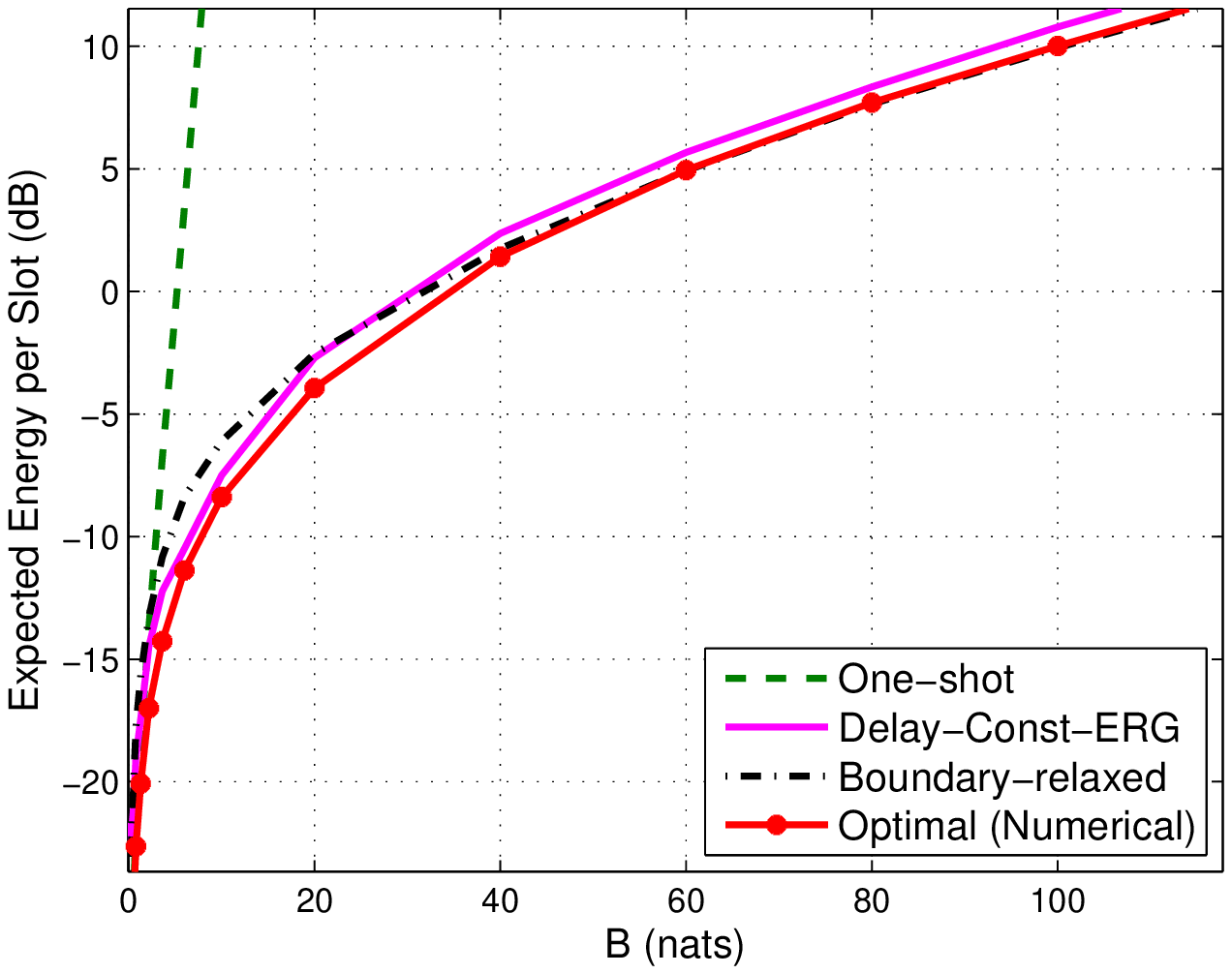}}
    \caption{Per slot energy cost for $T=5$ and $T=50$}
    \label{fig:simul_JT_one_delta_relaxed_Blin}
\end{figure}
the suboptimal schedulers is plotted for $T=5$ and $T=50$ assuming that the fading $\{g_t\}_{t=1}^T$ are i.i.d.~truncated exponential with a support of $[0.001, 10^6]$, i.e.,
\begin{equation} \label{eq:trunc_pdf}
    f(g) = \begin{cases}
        ce^{-(g-0.001)},& \text{if}\;\; 0.001\le g \le 10^6,\\
        0,& \text{otherwise},
    \end{cases}
\end{equation}
where $c$ is a normalization factor.
As can be seen, the one-shot scheduler is near-optimal only when $B$ is small. The other schedulers performs close to the optimal through all ranges of $B$. When $T=5$, as in Fig.~\ref{fig:simul_JT_one_delta_relaxed_Blin}a, the delay-constrained ergodic scheduler performs worse than the boundary-relaxed for all $B$. This is because $T=5$ is too small for the delay-constrained ergodic scheduler to perform like the optimal. When $T=50$, as in Fig.~\ref{fig:simul_JT_one_delta_relaxed_Blin}b, there exists a range of $B$ such that the delay-constrained ergodic scheduler outperforms the boundary-relaxed scheduler.  This phenomenon can be clearly illustrated in Fig.~\ref{fig:simul_J50_one_delta_relaxed_Blog}, where the number of bits are given in logarithmic scale.
\begin{figure}
    \centering
    \includegraphics[width=0.6\textwidth]{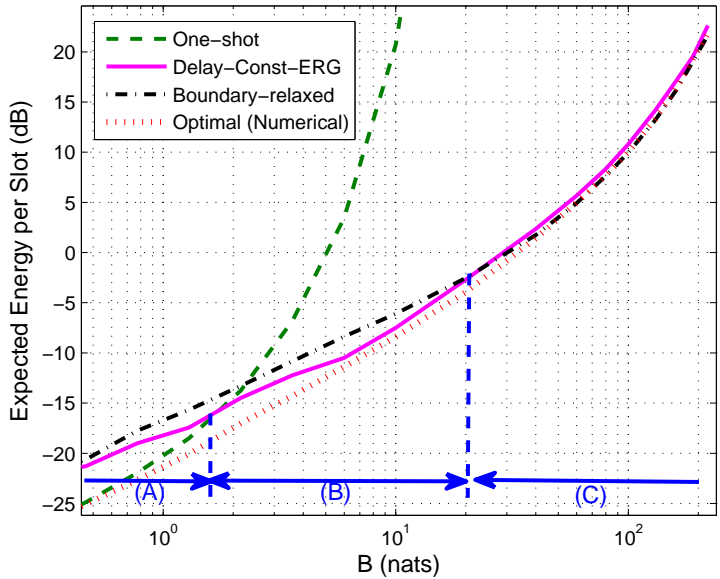}
    \caption{Average energy cost per slot for $T=50$ when $g$ is a truncated exponential variable with support $[0.001, 10^6]$}
    \label{fig:simul_J50_one_delta_relaxed_Blog}
\end{figure}
As can be seen in Fig.~\ref{fig:simul_J50_one_delta_relaxed_Blog}, the one-shot scheduler performs best for small $B$ (region $\mathsf{A}$) and the boundary-relaxed scheduler outperforms when $B$ is very large (region $\mathsf{C}$). In the middle range (region $\mathsf{B}$), the delay-constrained ergodic scheduler performs better than the other two.

\vspace{20pt}
\section{Scheduling Gain}

We have shown that the boundary-relaxed and the one-shot schedulers are asymptotically optimal as $B\to\infty$ and $B\to 0$, respectively. Another interesting issue is quantifying the advantage these schedulers provide compared to a non-opportunistic equal-bit scheduler that simply transmits $B/T$ bits during each time slot.

To compare energy performance, we first calculate the expected energy cost of the equal-bit scheduler, which is
\begin{equation} \label{eq:JT_eq}
        \bar{J}_T^\text{equal}(B) = \E\left[\sum_{t=1}^T \frac{e^{\frac{B}{T}}-1}{g_t}\right] = T\left(e^{\frac{B}{T}} \nu_1-\nu_1\right),
\end{equation}
since the equal-bit scheduler chooses $b_t=B/T$ for all $t$. Notice that the equal-bit scheduler achieves the delay-limited capacity \cite{Caire_IT99} \cite{Hanly_IT98} (i.e., zero-outage capacity) with rate $B/T$.

We define the scheduling gain as the ratio between the expected energy expenditures:
\begin{equation} \label{eq:scheduling_gain_opt}
    \Delta_T^\text{opt}(B) \triangleq \frac{\bar{J}_T^\text{equal}(B)}{\bar{J}_T^\text{opt}(B)}
\end{equation}
and quantify its behavior in the following theorem:
\vspace{10pt}
\begin{theorem} \label{thm:scheduling_gain_bounds}
    For any $T$, the scheduling gain $\Delta_T^{\rm opt}(B)$ is monotonically decreasing with respect to $B$. Furthermore, the limiting scheduling gains are given by:
    \begin{equation}
        \lim_{B\to 0} \Delta_T^{\rm opt}(B) = \lim_{B\to 0} \frac{\bar{J}_T^{\rm equal}(B)}{\bar{J}_T^{\rm one}(B)} = \frac{\nu_1}{\omega_{T+1}}, \label{eq:Delta_smallB}
    \end{equation}
    and if the PDF of the fading distribution is compactly supported and continuous,
    \begin{equation}
        \lim_{B\to\infty}\Delta_T^{\rm opt}(B) = \lim_{B\to\infty} \frac{\bar{J}_T^{\rm equal}(B)}{\bar{J}_T^{\rm relax}(B)} = \frac{\nu_1}{\G(\nu_T,\cdots, \nu_1)}. \label{eq:Delta_largeB}
    \end{equation}
\end{theorem}
\vspace{10pt}
\begin{proof}
    See Appendix \ref{sec:pf_scheduling_gain_bounds}.
\end{proof}
\vspace{10pt}
Since the boundary-relaxed scheduler is optimal as $B\to\infty$, the scheduling gain of the optimal scheduler and that of the boundary-relaxed scheduler are the same as $B\to\infty$; the same is true for the optimal and the one-shot scheduler as $B\to 0$.
The plot of scheduling gain vs.~$B$ in Fig.~\ref{fig:gain_T5} agrees with the results of Theorem \ref{thm:scheduling_gain_bounds}.
Intuitively, scheduling delivers a larger power gain for small $B$ because in such scenarios one can be very opportunistic and transmit the entire packet once a sufficiently good channel state is realized. For larger $B$, however, it is inefficient to transmit the entire packet in a single slot (because energy increases exponentially with the number of bits) and thus transmissions must be spread across many slots (in fact, all slots are used as $B\to \infty$), which reduces the channel quality during those transmissions and thus reduces the benefit of scheduling.

In Table \ref{tbl:scheduling_gain} the limited scheduling gains are showed for various fading distributions. As intuitively expected, the scheduling gain is larger for more severe fading distributions and for larger time horizons $T$. From the fact that both $\G(\nu_T, \cdots, \nu_1)$ and $\omega_{T+1}$ decrease as $T$ increases \cite{Lee_WC09}, the asymptotic scheduling gains in \eqref{eq:Delta_smallB} and \eqref{eq:Delta_largeB} increase with $T$. 

\begin{figure}
    \centering
    \includegraphics[width=0.6\textwidth]{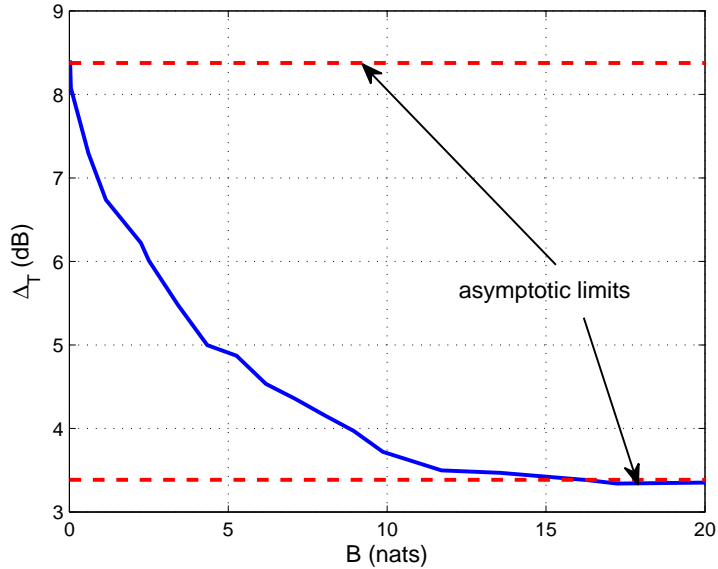}
    \caption{Scheduling gain $\Delta_5$ when $g$ is a truncated exponential variable with support $[0.001, 10^6]$}
    \label{fig:gain_T5}
\end{figure}

\begin{table*}
    \centering
    \caption{Scheduling gain examples for several fading distributions} \label{tbl:scheduling_gain}
\begin{tabulary}{50pt}{|l||p{1.6cm}|p{1.6cm}|p{1.6cm}|p{1.6cm}|p{1.6cm}|p{1.6cm}|}
\hline
& \multicolumn{2}{c|}{$T=5$} & \multicolumn{2}{c|}{$T=10$} & \multicolumn{2}{c|}{$T=50$} \\ \cline{2-3} \cline{4-5} \cline{6-7}
\raisebox{1.5ex}[0cm][0cm]{distribution of channel state $g$} & $\lim\limits_{B\to\infty}\Delta_5^\text{opt}(B)$ & $\lim\limits_{B\to 0}\Delta_5^\text{opt}(B)$ & $\lim\limits_{B\to \infty}\Delta_{10}^\text{opt}(B)$ & $\lim\limits_{B\to 0}\Delta_{10}^\text{opt}(B)$ & $\lim\limits_{B\to \infty}\Delta_{50}^\text{opt}(B)$ & $\lim\limits_{B\to 0}\Delta_{50}^\text{opt}(B)$ \\ \hline
truncated exponential with supp.~$[0.1, 10^6]$ & 0.97 dB & 4.42 dB & 1.26 dB & 5.98 dB & 1.63 dB & 8.59 dB  \\ \hline
truncated exponential with supp.~$[0.01, 10^6]$ & 2.19 dB & 6.72 dB & 2.80 dB & 8.63 dB & 3.52 dB & 11.51 dB \\ \hline
truncated exponential with supp.~$[0.001, 10^6]$ & 3.38 dB & 8.38 dB & 4.22 dB & 10.44 dB  & 5.17 dB & 13.40 dB \\ \hline
\end{tabulary}
\end{table*}

\subsection{Large $B$ Behavior (High SNR)}
\label{sec:largeB}

When $B$ is large relative to $T$, it is useful to interpret the scheduling gain in terms of the well-known affine approximation to high-SNR ($P$) capacity \cite{Shamai_IT01}:
    $C(P)=\mathcal{S}_\infty (\log P -\mathcal{L}_\infty) + o(1)$,
where $\mathcal{S}_\infty$ denotes the slope representing the multiplexing gain and $\mathcal{L}_\infty$ denotes the constant term representing the power/rate offset. We define the average SNR on a per-slot basis, i.e., $P=\bar{J_T}/T$. Similarly, the average rate is defined as $R_T = B/T$, which represents the average spectral efficiency per slot. Then, we investigate $R_T$ in terms of $P$ and $T$:
\begin{equation} \label{eq:high_affine}
    R_T(P)=\mathcal{S}_\infty (\log P-\mathcal{L}_{\infty, T}) +o(1).
\end{equation}

With algebraic calculations, we can obtain $\mathcal{S}_\infty$ and $\mathcal{L}_{\infty, T}$ for the equal-bit policy, the optimal scheduler (which is equal to the boundary-relaxed scheduler in this regime\footnote{We obtain this result in the process of proving Theorem \ref{thm:relax_cost_conv}, and thus we limit the fading distribution as conditioned in Theorem \ref{thm:relax_cost_conv}, i.e., the PDF $f$ is compactly supported and is continuous on the support.}), as well as the ergodic capacity (see Appendix \ref{sec:deriv_high_snr_approx} for derivation).
The three policies have the same multiplexing gain (degrees of freedom) per slot ($\mathcal{S}_\infty=1$), but the offsets $\mathcal{L}_{\infty,T}$ are different:
\begin{eqnarray}
    \mathcal{L}_{\infty,T}^\text{equal} &=& \log \nu_1, \\
    \mathcal{L}_{\infty,T}^\text{opt} &=& \log \G\left(\nu_T, \nu_{T-1},\cdots,\nu_1\right), \\
    \mathcal{L}_{\infty,T}^\text{erg} &=& \log \nu_\infty.
\end{eqnarray}
The offset of the equal-bit scheduler is independent of $T$ since it does not take advantage of time diversity. On the other hand, the offset of the boundary-relaxed scheduler decreases with $T$ since $\G\left(\nu_T, \nu_{T-1},\cdots, \nu_1\right)$ decreases \cite{Lee_WC09}. Moreover, the offset of the boundary-relaxed scheduler converges to that of the ergodic capacity because $\G\left(\nu_T, \nu_{T-1},\cdots, \nu_1\right)\to \nu_\infty$ as $T\to\infty$ \cite{Lee_WC09}.
Figure \ref{fig:Linf_vs_T} illustrates the offsets $\mathcal{L}_{\infty,T}$ for several fading distributions. As can be seen, $\mathcal{L}_{\infty,T}^\text{opt}$ for all the fading distributions decreases from $\mathcal{L}_{\infty,T}^\text{equal}$ as $T$ increases and converges to $\mathcal{L}_{\infty,T}^\text{erg}$. We can also see that the offsets $\mathcal{L}_{\infty,T}$ have larger values for more severe fading distributions.
\begin{figure}
    \centering
    \includegraphics[width=0.6\textwidth]{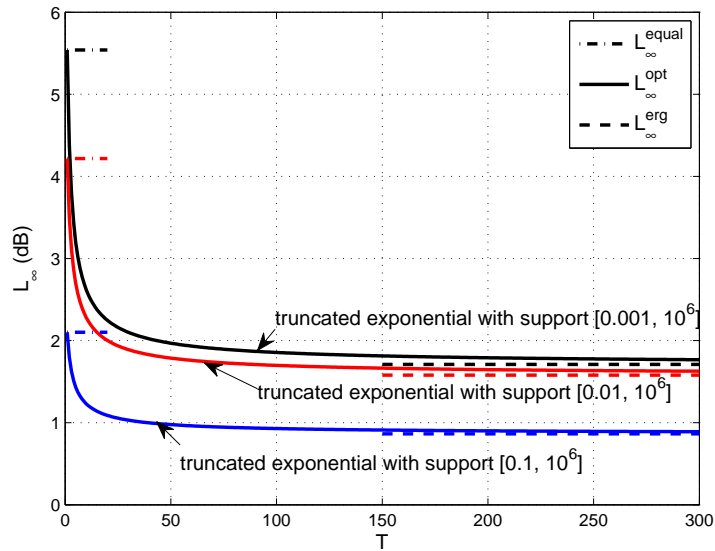}
    \caption{$\mathcal{L}_{\infty,T}$ for several fading distributions}
    \label{fig:Linf_vs_T}
\end{figure}

Figure \ref{fig:high_SNR_approx} illustrates the behavior of the spectral efficiency versus SNR. The dashed lines are obtained from the affine approximations in \eqref{eq:high_affine} while the solid lines are obtained numerically by running the optimal scheduling policy. As can be seen, the affine approximations are very accurate when SNR is 20dB or higher.
\begin{figure}
    \centering
    \includegraphics[width=0.6\textwidth]{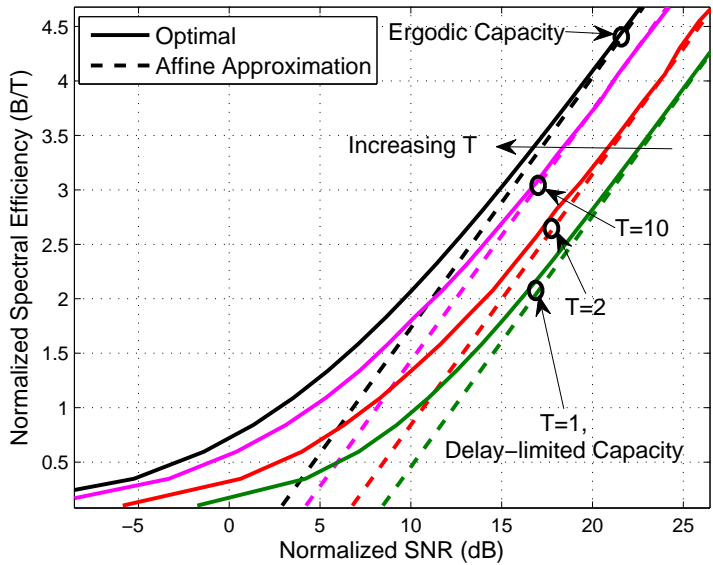}
    \caption{High SNR behavior when  $g$ is a truncated exponential variable with support $[0.001, 10^6]$}
    \label{fig:high_SNR_approx}
\end{figure}
Furthermore, as $T$ increases the spectral efficiency increases from the delay-limited capacity (achieved with the $T=1$ optimal scheduling or the equal-bit scheduling) to the ergodic capacity (achieved with the $T=\infty$ optimal scheduling).

It is interesting to note that for the dual problem of rate maximization over a finite time horizon when subject to a per-realization energy constraint (i.e., for every realization of channel gains $g_T, \cdots, g_1$ the amount of energy used by the scheduling policy cannot exceed some constraint $E$ considered in \cite{Negi_IT02} and \cite{Caire_IT04}, at high SNR the optimal policy converges to uniform power allocation (independent of channel state) and there is no advantage to using an intelligent scheduling policy. This is to be contrasted with the setting considered here, where there is a non-vanishing benefit to using the optimal scheduler even at high SNR (i.e., large $B$).

\subsection{Small $B$ Behavior (Low SNR)}

In this regime, we characterize the linear approximation to the spectral efficiency versus $\frac{E_b}{N_0}$ curve based on the wideband analysis in \cite{Verdu_IT02}.
The linear approximation consists of a constant term $\left(\frac{E_b}{N_0}\right)_{\min}$ and a slope $\mathcal{S}_0$ that represent the minimum energy per bit for reliable communication and the growth of spectral efficiency with respect to $\frac{E_b}{N_0}$. To be clear, we adopt the notion of $E_b$ as the required energy per slot to transmit one bit per slot instead of the required energy to transmit one bit throughout the entire $T$ slots:
\begin{equation}
    R_T\left(\frac{E_b}{N_0}\right) \approx \mathcal{S}_{0,T} \left(\frac{\frac{E_b}{N_0}\Big\vert_{\text{dB}}-\left(\frac{E_b}{N_0}\right)_{\min,T}\Big\vert_{\text{dB}}}{3\; \text{dB}}\right).
\end{equation}

These parameters $\mathcal{S}_{0,T}$ and $\left(\frac{E_b}{N_0}\right)_{\min,T}$ can be obtained for the equal-bit scheduler and the one-shot scheduler, which is optimal for $B\to 0$, (see Appendix \ref{sec:deriv_low_snr_approx} for derivations):
\begin{eqnarray}
    &\mathcal{S}_{0,T}^\text{eq} = 2, & \left(\frac{E_b}{N_0}\right)_{\min,T}^\text{eq} = (\log 2)\nu_1,\\
    &\mathcal{S}_{0,T}^\text{one} = \frac{2}{T}, & \left(\frac{E_b}{N_0}\right)_{\min,T}^\text{one} = (\log 2) \E\left[\min\left(\frac{1}{g_T}, \E\left[\min\left(\frac{1}{g_{T-1}}, \cdots \E\left[\min\left(\frac{1}{g_2},\E\left[\frac{1}{g_1}\right]\right)\right] \right)\right]\right)\right],
\end{eqnarray}
and both $\mathcal{S}_{0,T}^\text{erg}$ and $\left(\frac{E_b}{N_0}\right)_{\min,T}^\text{erg}$ are zero for ergodic capacity.

Figure \ref{fig:low_SNR_approx} illustrates the behavior of $\left(\frac{E_b}{N_0}\right)_{\min,T}$ and $\mathcal{S}_{0,T}$ with respect to $T$. As can be seen, both $\left(\frac{E_b}{N_0}\right)_{\min,T}$ and $\mathcal{S}_{0,T}$ decrease from the  delay-limited values to the ergodic capacity values as $T\to\infty$ due to the available time diversity.
\begin{figure}
    \centering
    \subfloat[$\left(\frac{E_b}{N_0}\right)_{\min,T}$ versus $T$]{\includegraphics[width=0.48\textwidth]{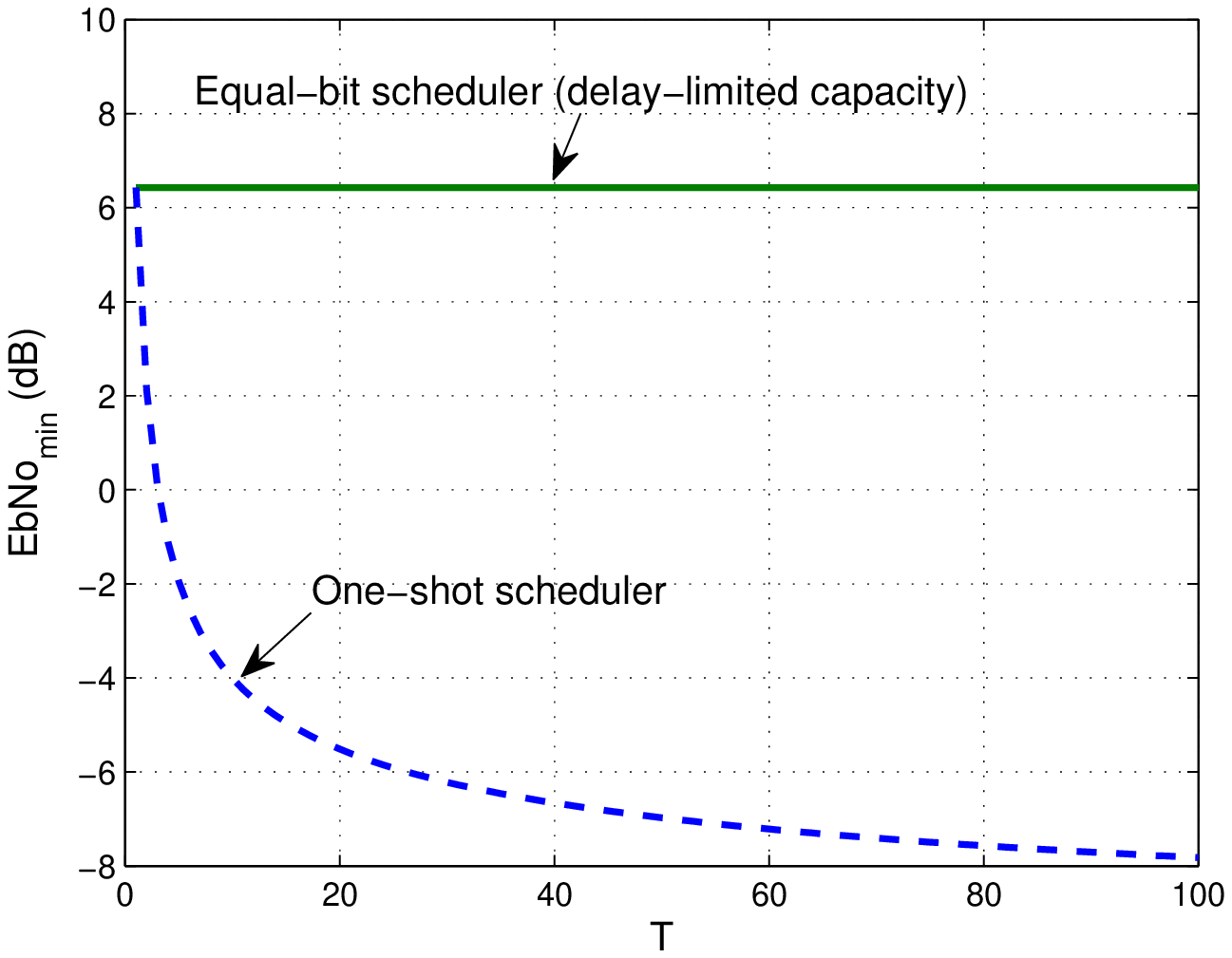}}\
    \hfill
    \subfloat[$\mathcal{S}_{0,T}$ versus $T$]{\includegraphics[width=0.48\textwidth]{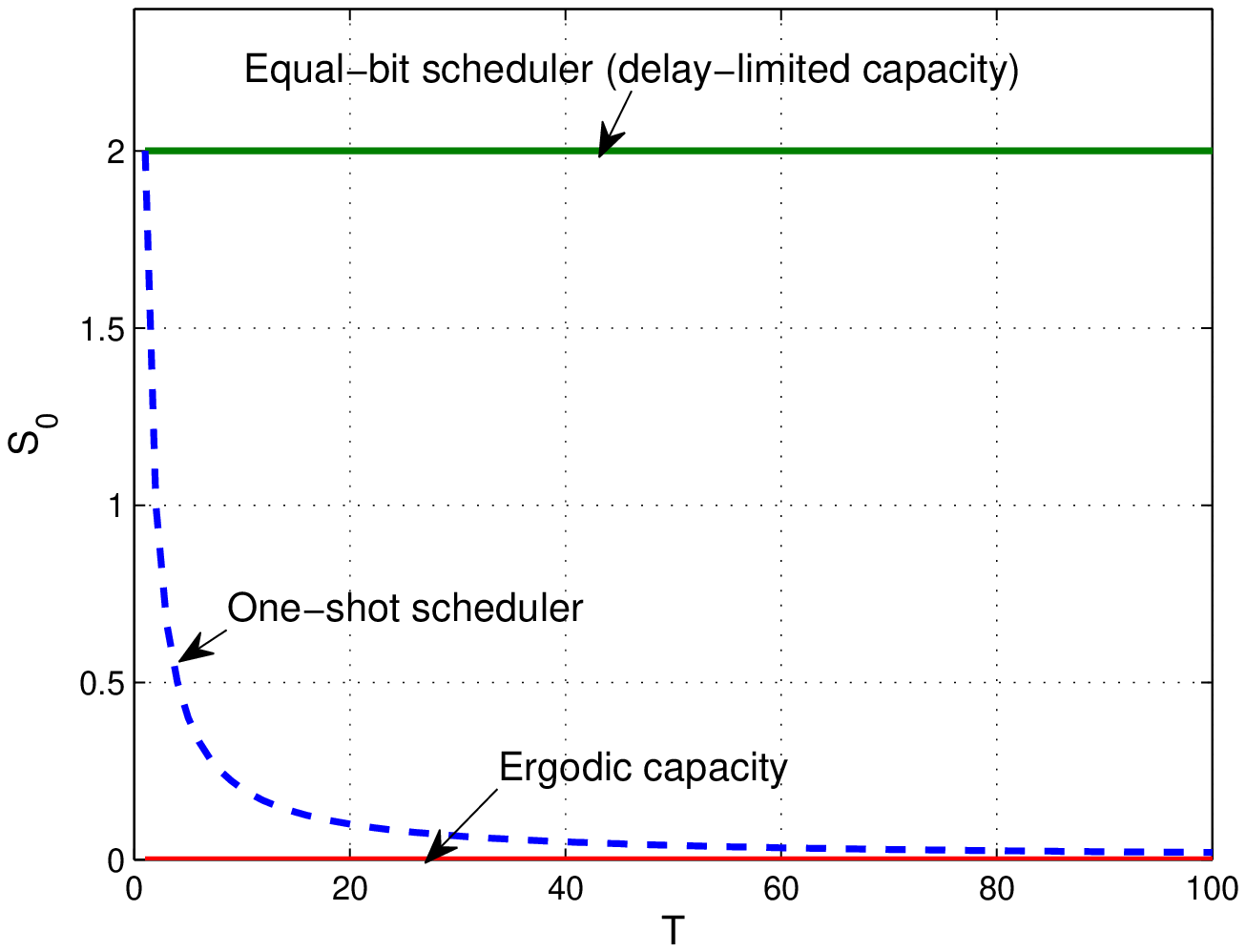}}
    \caption{Low SNR behavior when  $g$ is a truncated exponential variable with support $[0.001, 10^6]$}
    \label{fig:low_SNR_approx}
\end{figure}

\section{Conclusion}

We have shown the asymptotic optimality of three different scheduling policies for delay-constrained transmission over a fading channel. When only a small number of bits need to be served, a one-shot threshold
policy is optimal: once a sufficiently good channel state is experienced, the entire packet is transmitted. On the other hand, when the number of bits is large, the number of transmitted bits at each time step should be a weighted sum of the unserved bits and a channel state-related term, where the weight is proportional to the time to deadline. In each of these two policies, the scheduler is opportunistic while also being cognizant of the deadline. Furthermore, a modification of the ergodic
waterfilling policy is shown to be optimal when the number of bits and the time horizon are both large.

Although problems involving delay-limited communication are of great practical importance and have been the subject of considerable research, such problems generally do not have closed-form solutions.  In this work, however, we are able to circumvent this general difficulty by considering different asymptotic regimes.  It would be interesting to see if the asymptotically optimal policies identified here, which admit a very simple analytical form, can be extended to other more general settings.  For example, to scheduling with time-varying channels and randomly arriving packets \cite{Berry_IT02}\cite{Rajan_IT04} and possibly to multi-user channels \cite{Yeh_ISIT03}.

\appendices

\section{Proof of Proposition \ref{prop:monotonicity_bt_opt}}
\label{sec:pf_monotonicity_bt_opt}
    \begin{enumerate}
        \item[(a)]
        \begin{enumerate}
            \item[(i)] We show the monotonicity of $b_t^\text{opt}$ in $\beta_t$. When $g_t\le \frac{1}{(\bar{J}_{t-1}^\text{opt})'(\beta_t)}$, $b_t^\text{opt}=0$ and thus is non-decreasing in $\beta_t$. When $\frac{1}{(\bar{J}_{t-1}^\text{opt})'(\beta_t)}< g_t < \frac{e^{\beta_t}}{(\bar{J}_{t-1}^\text{opt})'(0)}$, we suppose that $b_t^\text{opt}$ decreases in $\beta_t$. Then, $(\beta_t-b_t^\text{opt})$ increases and $(\bar{J}_{t-1}^\text{opt})'(\beta_t-b_t^\text{opt})$ increases. As a result, $b_t^\text{opt}$ increases but this leads a contradiction. Thus, $b_t^\text{opt}$ is non-decreasing in $\beta_t$ when $\frac{1}{(\bar{J}_{t-1}^\text{opt})'(\beta_t)}< g_t < \frac{e^{\beta_t}}{(\bar{J}_{t-1}^\text{opt})'(0)}$. When $g_t \ge \frac{e^{\beta_t}}{(\bar{J}_{t-1}^\text{opt})'(0)}$, $b_t^\text{opt}=\beta_t$ and thus is non-decreasing in $\beta_t$.

            \item[(ii)] We show the strict monotonicity of $b_t^\text{opt}$ for large $\beta_t$. To do this, we first show the unboundedness of $(\bar{J}_{t-1}^\text{opt})'$. Suppose not, i.e., there exists $M (<\infty)$ such that $(\bar{J}_{t-1}^\text{opt})'(\beta)\le M$ for all $\beta(\ge 0)$. By integrating both sides, we have $\bar{J}_{t-1}^\text{opt}(\beta)\le M\beta$ for all $\beta(\ge 0)$. Note also that $\bar{J}_{t-1}^\text{opt}(\beta)\ge \E\left[\frac{e^{\beta/t}-1}{\max(g_{t-1},\cdots,g_1)}\right]$ for all $\beta(\ge 0)$ by \eqref{eq:Jt_opt}. Consequently, we have
$\E\left[\frac{e^{\beta/t}-1}{\max(g_{t-1},\cdots,g_1)}\right] \le M\beta$ for all $\beta(\ge 0)$, which leads a contradiction. Therefore,  $(\bar{J}_{t-1}^\text{opt})'$ is unbounded.

            Since $(\bar{J}_{t-1}^\text{opt})'$ is unbounded and monotonically increasing, for any given $g_t$ there exists $\mathfrak{B}_0$ such that $\frac{1}{(\bar{J}_{t-1}^\text{opt})'(\beta_t)}< g_t < \frac{e^{\beta_t}}{(\bar{J}_{t-1}^\text{opt})'(0)}$ for all $\beta_t> \mathfrak{B}_0$. In this region of $g_t$, we showed that $b_t^\text{opt}$ is non-decreasing in $\beta_t$ by (i). Suppose $b_t^\text{opt}$ maintains a constant value as $\beta_t$ increases in this region of $g_t$. Then, $(\beta_t-b_t^\text{opt})$ increases and $(\bar{J}_{t-1}^\text{opt})'(\beta_t-b_t^\text{opt})$ increases. As a result, $b_t^\text{opt}$ increases but this leads a contradiction, too. Therefore, $b_t^\text{opt}$ increases strictly in $\beta_t$ if $\beta_t> \mathfrak{B}_0$.

            \item[(iii)] Finally, we show the monotonicity of $(\beta_t-b_t^\text{opt})$ in $\beta_t$. Since $b_t^\text{opt}$ is non-decreasing in $\beta_t$, $\frac{e^{b_t^\text{opt}}}{g_t}$ is non-decreasing. Since $(\bar{J}_{t-1}^\text{opt})'(\cdot)$ is an increasing function, $(\beta_t-b_t^\text{opt})$ must be non-decreasing in $\beta_t$. If $\beta_t>\mathfrak{B}_0$, then $b_t^\text{opt}$ is strictly increasing by (ii) and thus $(\beta_t-b_t^\text{opt})$ is strictly increasing by the same argument.
        \end{enumerate}

        \item[(b)] If $g_t \le \frac{1}{(\bar{J}_{t-1}^\text{opt})'(\beta_t)}$ and $g_t \ge \frac{e^{\beta_t}}{(\bar{J}_{t-1}^\text{opt})'(0)}$, $b_t^\text{opt}$ is constant as $g_t$ increases, and thus is non-decreasing with respect to $g_t$. When $\frac{1}{(\bar{J}_{t-1}^\text{opt})'(\beta_t)} < g_t < \frac{e^{\beta_t}}{(\bar{J}_{t-1}^\text{opt})'(0)}$,
        \begin{equation} \label{eq:bt_logged}
            b_t^\text{opt} = \log g_t + \log \left[(\bar{J}_{t-1}^\text{opt})'(\beta_t-b_t^\text{opt})\right].
        \end{equation}
        If we suppose that $b_t^\text{opt}$ decreases strictly as $g_t$ increases, $(\beta_t-b_t^\text{opt})$ will increase and thus $(\bar{J}_{t-1}^\text{opt})'(\beta_t-b_t^\text{opt})$ will also increase. This leads a contradiction because the left hand side of \eqref{eq:bt_logged} decreases strictly while the right hand side increases. Therefore, $b_t^\text{opt}$ is non-decreasing in $g_t$.
    \end{enumerate}

\section{Proof of Theorem \ref{thm:relax_policy_conv}}
\label{sec:pf_relax_policy_conv}
    We show the result by induction, i.e., we show that if the scheduling functions converge at time step $t-1$, then the functions also converge at time step $t$.
    The base cases occur at $t=1$ and $t=2$: by construction, $b_2^\text{relax}(\beta,g_2)= b_2^\text{opt}(\beta,g_2)$ for every $(\beta,g_2)$, and $b_1^\text{relax}(\beta,g_1) =b_1^\text{opt}(\beta,g_1)$ for every $(\beta,g_1)$.

        In order to show policy convergence, it is useful to write $b_t^\text{relax}$ as:
    \begin{equation} \label{eq:bt_relax_det}
         b_t^\text{relax} (\beta_t, g_t)=
        \begin{cases}
        0, & g_t \le \frac{1}{(\bar{U}_{t-1})'(\beta_t)},\\
        \arg_b \left\{\frac{e^{b}}{g_t} = (\bar{U}_{t-1})'(\beta_t-b)\right\}, & \hspace{-10pt}         \frac{1}{(\bar{U}_{t-1})'(\beta_t)} < g_t < \frac{e^{\beta_t}}{(\bar{U}_{t-1})'(0)}, \\
  \beta_t, & g_t \ge \frac{e^{\beta_t}}{(\bar{U}_{t-1})'(0)},
         \end{cases}
    \end{equation}
    which is identical to the expression for $b_t^\text{opt}$ in \eqref{eq:bt_opt_det} except replacing $(\bar{J}_{t-1}^\text{opt})'$ with $(\bar{U}_{t-1})'$.
    Since $\bar{J}_{t-1}^\text{opt}$ and $\bar{U}_{t-1}$ are convex, $(\bar{J}_{t-1}^\text{opt})'$ and $(\bar{U}_{t-1})'$ are increasing and moreover unbounded (shown in Appendix \ref{sec:pf_monotonicity_bt_opt} (a)(ii)).
    Since $g_{\min}>0$ and $g_{\max}<\infty$ where $g_{\min}$ and $g_{\max}$ are the lower and the upper bounds of the support of the PDF $f$ ($\text{Support}(f)=[g_{\min}, g_{\max}]$), there exists $\mathfrak{B}_0$ such that if $\beta>\mathfrak{B}_0$ then
    \begin{eqnarray}
        g_{\min}&>& \max\left(\frac{1}{(\bar{U}_{t-1})'(\beta)}, \; \frac{1}{(\bar{J}_{t-1}^\text{opt})'(\beta)} \right) \\
        g_{\max}&<& \min\left(\frac{e^\beta}{(\bar{U}_{t-1})'(0)},\; \frac{e^\beta}{(\bar{J}_{t-1}^\text{opt})'(0)} \right).
    \end{eqnarray}
    Henceforth, we only consider $\beta>\mathfrak{B}_0$, and thus, no truncation occurs in both policy functions, i.e., $b_t^\text{relax}$ and $b_t^\text{opt}$ are determined by
    \begin{eqnarray}
        \frac{e^{b_t^\text{relax}}}{g_t} &=& \phi(\beta-b_t^\text{relax}), \label{eq:ebrelax_phi} \\
        \frac{e^{b_t^\text{opt}}}{g_t} &=& \psi(\beta-b_t^\text{opt}), \label{eq:ebopt_psi}
    \end{eqnarray}
    for $\beta>\mathfrak{B}_0$, where $\phi(\beta) = (\bar{U}_{t-1})'(\beta)$ and $\psi(\beta) = (\bar{J}_{t-1}^\text{opt})'(\beta)$.

    Let $\epsilon>0$ be given. By Lemma \ref{lem:bt_conv_Ud_Jd} (stated later in Appendix \ref{sec:pf_relax_policy_conv}), there exists $\mathfrak{B}_1(\ge \mathfrak{B}_0)$ such that if $\xi>\mathfrak{B}_1$, then
    \begin{equation}
        \left|\phi(\xi)-\psi(\xi)\right| < \epsilon.
    \end{equation}
    Since  $\beta-b_t^\text{opt}$ and $b_t^\text{opt}$ are strictly increasing in $\beta$ (when $\beta$ is sufficiently large) by Proposition \ref{prop:monotonicity_bt_opt} and $\beta-b_t^\text{opt}(=\psi^{-1}(e^{b_t^\text{opt}}/g_t))$ is unbounded due to the unboundedness and the monotonicity of $\psi$,
 there exists $\mathfrak{B}_2$ such that $\beta>\mathfrak{B}_2$ implies
    \begin{equation}
        \beta-b_t^\text{opt}(\beta,g_t) > \mathfrak{B}_1,\qquad \forall g_t\in [g_{\min}, g_{\max}].
    \end{equation}
    Therefore, if $\beta> \mathfrak{B}_2$,
        \begin{equation} \label{eq:phi_psi_same}
        \left|\phi(\beta-b_t^\text{opt}(\beta, g_t))-\psi(\beta-b_t^\text{opt}(\beta,g_t))\right| < \epsilon, \quad \forall g_t\in [g_{\min}, g_{\max}].
    \end{equation}
    If $b_t^\text{opt}\le b_t^\text{relax}$,
    \begin{equation} \label{eq:phi_psi_diff_1}
        \phi(\beta-b_t^\text{relax}) \le \phi(\beta-b_t^\text{opt}) < \psi (\beta-b_t^\text{opt})+\epsilon,
    \end{equation}
    where the last inequality follows from \eqref{eq:phi_psi_same}. Additionally, we have
    \begin{equation} \label{eq:phi_psi_diff_2}
        \phi(\beta-b_t^\text{relax})=\frac{e^{b_t^\text{relax}}}{g_t} \ge \frac{e^{b_t^\text{opt}}}{g_t} = \psi(\beta-b_t^\text{opt}).
    \end{equation}
    Combining \eqref{eq:phi_psi_diff_1} and \eqref{eq:phi_psi_diff_2}, we have $\phi(\beta-b_t^\text{relax})-\psi(\beta-b_t^\text{opt})<\epsilon$. By the same argument for $b_t^\text{opt}>b_t^\text{relax}$, we have $\psi(\beta-b_t^\text{opt})-\phi(\beta-b_t^\text{relax})<\epsilon$. Thus, we obtain
    \begin{equation}
        \left|\phi(\beta-b_t^\text{relax}(\beta,g_t))-\psi(\beta-b_t^\text{opt}(\beta,g_t))\right| < \epsilon,\qquad \forall g_t\in [g_{\min}, g_{\max}].
    \end{equation}
    By \eqref{eq:ebrelax_phi}, \eqref{eq:ebopt_psi}, and the continuity, $b_t^\text{relax}(\cdot, g_t) \to b_t^\text{opt}(\cdot, g_t)$ uniformly on $[g_{\min}, g_{\max}]$ is obtained.
\hfill\QED

\vspace{20pt}
\begin{lemma} \label{lem:bt_conv_Ud_Jd}
        If $b_{t-1}^\text{relax}(\cdot,g_{t-1}) \to b_{t-1}^\text{opt}(\cdot,g_{t-1})$ uniformly on $[g_{\min}, g_{\max}]$, then
        \begin{equation}
            \lim_{\beta\to\infty} \left[(\bar{U}_{t-1})'(\beta)-(\bar{J}_{t-1}^\text{opt})'(\beta) \right]=0.
        \end{equation}
    \end{lemma}
    \begin{proof}
    From \eqref{eq:Jt_opt} and \eqref{eq:bt_opt_det}, we write the expected cost-to-go as:
    \begin{equation}
        \bar{J}_{t-1}^\text{opt}(\beta)
        = \int_0^{\frac{1}{(\bar{J}_{t-2}^\text{opt})'(\beta)}} \bar{J}_{t-2}^\text{opt}(\beta) dF(x)
        + \int_{\frac{e^\beta}{(\bar{J}_{t-2}^\text{opt})'(0)}}^\infty \frac{e^\beta-1}{x} dF(x)
        + \int_{\frac{1}{(\bar{J}_{t-2}^\text{opt})'(\beta)}}^{\frac{e^\beta}{(\bar{J}_{t-2}^\text{opt})'(0)}}
        \left[\frac{e^{b_{t-1}^\text{opt}}-1}{x} + \bar{J}_{t-2}^\text{opt}(\beta-b_{t-1}^\text{opt}) \right] dF(x),
    \end{equation}
    where $b_{t-1}^\text{opt}$ is a function of $\beta$ (and $x$).
    By differentiating $\bar{J}_{t-1}^\text{opt}$ using integral calculus\footnote{ $H'(x)=h(x,\varphi(x))\varphi'(x)+\int_a^{\varphi(x)} \frac{\partial h}{\partial x}(x,y) dy$ for $H(x)=\int_a^{\varphi(x)} h(x,y) dy$}, the derivative (with respect to $\beta$) of $\bar{J}_{t-1}^\text{opt}$ is:
    \begin{equation} \label{eq:Jt1d}
        (\bar{J}_{t-1}^\text{opt})'(\beta) =\int_0^{\frac{1}{(\bar{J}_{t-2}^\text{opt})'(\beta)}} (\bar{J}_{t-2}^\text{opt})'(\beta) f(x) dx
        + e^\beta \int_{\frac{e^\beta}{(\bar{J}_{t-2}^\text{opt})'(0)}}^\infty \frac{1}{x} f(x)dx
        + \int_{\frac{1}{(\bar{J}_{t-2}^\text{opt})'(\beta)}}^{\frac{e^\beta}{(\bar{J}_{t-2}^\text{opt})'(0)}}
        \frac{d}{d\beta}\left[\frac{e^{b_{t-1}^\text{opt}}-1}{x}+\bar{J}_{t-2}^\text{opt}(\beta-b_{t-1}^\text{opt})\right] dF(x)
    \end{equation}
    Since  $(\bar{J}_{t-2}^\text{opt})'$ is unbounded increasing and $\text{Support}(f)=[g_{\min}, g_{\max}]$, $\frac{1}{(\bar{J}_{t-2}^\text{opt})'(\beta)}< g_{\min}$ and $\frac{e^\beta}{(\bar{J}_{t-2}^\text{opt})'(0)}>g_{\max}$ for sufficiently large $\beta$, and thus
    \begin{eqnarray}
        &&\lim_{\beta\to\infty} \int_0^{\frac{1}{(\bar{J}_{t-2}^\text{opt})'(\beta)}} (\bar{J}_{t-2}^\text{opt})'(\beta) f(x) dx=       0,\\
        &&\lim_{\beta\to\infty} e^\beta \int_{\frac{e^\beta}{(\bar{J}_{t-2}^\text{opt})'(0)}}^\infty \frac{1}{x} f(x) dx
        = 0.
    \end{eqnarray}
    Since $\frac{e^{b_{t-1}^\text{opt}}}{x}=(\bar{J}_{t-2}^\text{opt})'(\beta-b_{t-1}^\text{opt})$ for $x\in \left(\frac{1}{(\bar{J}_{t-2}^\text{opt})'(\beta)}, \frac{e^\beta}{(\bar{J}_{t-2}^\text{opt})'(0)}\right)$ by \eqref{eq:bt_opt_det},
    \begin{equation} \label{eq:Jt1d_sub}
        \frac{d}{d\beta}\left[\frac{e^{b_{t-1}^\text{opt}}-1}{x}+\bar{J}_{t-2}^\text{opt}(\beta-b_{t-1}^\text{opt})\right]
        = \frac{e^{b_{t-1}^\text{opt}}}{x}.
    \end{equation}
    As a result, the derivative of the expected cost-to-go can be stated simply in the limit of large $\beta$:
    \begin{equation}
        \lim_{\beta\to\infty} (\bar{J}_{t-1}^\text{opt})'(\beta)
        = \lim_{\beta\to\infty} \int_{\frac{1}{(\bar{J}_{t-2}^\text{opt})'(\beta)}}^{\frac{e^\beta}{(\bar{J}_{t-2}^\text{opt})'(0)}} \frac{e^{b_{t-1}^\text{opt}(\beta,x)}}{x} dF(x).
    \end{equation}
    From \eqref{eq:barUt}, we have
    \begin{equation}
        \begin{split}
            \lim_{\beta\to\infty} \Big[(\bar{U}_{t-1})'(\beta)&-(\bar{J}_{t-1}^\text{opt})'(\beta) \Big]
            =\lim_{\beta\to\infty} \left[e^{\frac{\beta}{t-1}}\G(\nu_{t-1},\cdots, \nu_1)-\int_{\frac{1}{(\bar{J}_{t-2}^\text{opt})'(\beta)}}^{\frac{e^\beta}{(\bar{J}_{t-2}^\text{opt})'(0)}} \frac{e^{b_{t-1}^\text{opt}}}{x} dF(x)\right] \\
            &= \lim_{\beta\to\infty} \Bigg[e^{\frac{\beta}{t-1}}\G(\nu_{t-1},\cdots, \nu_1)-\int_{\frac{1}{(\bar{J}_{t-2}^\text{opt})'(\beta)}}^{\frac{e^\beta}{(\bar{J}_{t-2}^\text{opt})'(0)}} \frac{e^{b_{t-1}^\text{relax}}}{x} dF(x) +\int_{\frac{1}{(\bar{J}_{t-2}^\text{opt})'(\beta)}}^{\frac{e^\beta}{(\bar{J}_{t-2}^\text{opt})'(0)}} \frac{e^{b_{t-1}^\text{relax}}-e^{b_{t-1}^\text{opt}}}{x} dF(x)\Bigg].
        \end{split}
    \end{equation}
    Since $\text{Support}(f)=[g_{\min},g_{\max}]\subset \left[\frac{1}{(\bar{J}_{t-2}^\text{opt})'(\beta)},\frac{e^\beta}{(\bar{J}_{t-2}^\text{opt})'(0)}\right]$ (for large $\beta$) and the induction hypothesis that $b_{t-1}^\text{relax}(\beta,x)$ converges to $b_{t-1}^\text{opt}(\beta, x)$ uniformly on $x\in [g_{\min}, g_{\max}]$,
    \begin{equation}
        \lim_{\beta\to\infty} \int_{\frac{1}{(\bar{J}_{t-2}^\text{opt})'(\beta)}}^{\frac{e^\beta}{(\bar{J}_{t-2}^\text{opt})'(0)}} \frac{e^{b_{t-1}^\text{relax}}-e^{b_{t-1}^\text{opt}}}{x} f(x) dx= \lim_{\beta\to\infty}\int_{g_{\min}}^{g_{\max}} \frac{e^{b_{t-1}^\text{relax}}-e^{b_{t-1}^\text{opt}}}{x} f(x) dx = 0,
    \end{equation}
    and thus
    \begin{equation}
        \lim_{\beta\to\infty} \left[(\bar{U}_{t-1})'(\beta)-(\bar{J}_{t-1}^\text{opt})'(\beta) \right]
        =\lim_{\beta\to\infty} \left[ e^{\frac{\beta}{t-1}}\G(\nu_{t-1},\cdots, \nu_1) -\int_{\frac{1}{(\bar{J}_{t-2}^\text{opt})'(\beta)}}^{\frac{e^\beta}{(\bar{J}_{t-2}^\text{opt})'(0)}}\frac{e^{b_{t-1}^\text{relax}}}{x} dF(x)\right].
    \end{equation}

    By substituting \eqref{eq:bt_relax} into $b_{t-1}^\text{relax}$ and re-writing $\G(\nu_{t-1},\cdots, \nu_1)$ as
    \begin{equation}
        \G(\nu_{t-1},\cdots, \nu_1) = \left(\G(\nu_{t-2},\cdots, \nu_1)\right)^{\frac{t-2}{t-1}} \int \left(\frac{1}{x}\right)^{\frac{1}{t-1}} dF(x),
    \end{equation}
    we have
        $\lim_{\beta\to\infty} \left[(\bar{U}_{t-1})'(\beta)-(\bar{J}_{t-1}^\text{opt})'(\beta) \right]=0.$
    \end{proof}

\section{Proof of Theorem \ref{thm:relax_cost_conv}}
\label{sec:pf_relax_cost_conv}
We will prove this by showing that $\lim_{B\to\infty}\left[\bar{J}_{T}^\text{relax}(B) - \bar{U}_{T}(B)\right]=0$ and $\lim_{B\to\infty} \left[\bar{U}_T(B) - \bar{J}_T^\text{opt}(B)\right]=0$.
\begin{enumerate}
    \item[(i)]  First, we show $\lim\limits_{B\to\infty}\left[\bar{J}_{T}^\text{relax}(B) - \bar{U}_{T}(B)\right]=0$ by induction. Notice that $J_{t}^\text{relax}(\beta,g) \ge U_{t}(\beta,g)$ for all values of $\beta$ and $g$ by the constructions \eqref{eq:Lt_opt} and \eqref{eq:Jt_relax}, and thus, $\left|J_{t}^\text{relax}(\beta,g) - U_{t}(\beta,g)\right|=J_{t}^\text{relax}(\beta,g) - U_{t}(\beta,g)$ and $\left|\bar{J}_{t}^\text{relax}(\beta) - \bar{U}_{t}(\beta)\right|=\bar{J}_{t}^\text{relax}(\beta) - \bar{U}_{t}(\beta)$.

        By \eqref{eq:barUt} and \eqref{eq:Jt_relax}, $\bar{J}_1^\text{relax}\equiv \bar{U}_1$. Let $\epsilon (>0)$ be given.  As an induction hypothesis, we assume that
        \begin{equation} \label{eq:Jrelax_Ut_hypo}
            \lim\limits_{B\to\infty}\left[\bar{J}_{T-1}^\text{relax}(B) - \bar{U}_{T-1}(B)\right]=0,
        \end{equation}
        i.e., there exists $\mathfrak{B}_0$ such that if $B>\mathfrak{B}_0$ then $\bar{J}_{T-1}^\text{relax}(B) - \bar{U}_{T-1}(B) <\epsilon$.

    To differentiate $b_t^\text{relax}$ and the solution to \eqref{eq:Lt_opt}, we let $b_t^\text{untruncated}$ be the solution to \eqref{eq:Lt_opt}, and thus the relation of the two is $b_t^\text{relax}(\beta_t,g_t)= \left\langle b_t^\text{untruncated}(\beta_t,g_t)\right\rangle_0^{\beta_t}$, where $b_t^\text{untruncated}(\beta_t,g_t)=\frac{1}{t}\beta_t +\frac{t-1}{t}\log\left(\frac{g_t}{\eta_t^\text{relax}}\right)$ by \eqref{eq:bt_relax_untruncated}.

    Notice that
    \begin{equation}
    \footnotesize
        \begin{split}
        \bar{J}_{T}^\text{relax}(B) &- \bar{U}_{T}(B)
        = \E\left[J_{T}^\text{relax}(B,g_T) - U_{T}(B,g_T)\right] \\
        &= \E\left[J_{T}^\text{relax}(B,g_T) - U_{T}(B,g_T) \Bigg\vert \frac{1}{(\bar{U}_{T-1})'(B)} < g_T < \frac{e^{B}}{(\bar{U}_{T-1})'(0)}\right]\Pr\left\{\frac{1}{(\bar{U}_{T-1})'(B)} < g_T < \frac{e^{B}}{(\bar{U}_{T-1})'(0)}\right\}\\
        &+ \E\left[J_{T}^\text{relax}(B,g_T) - U_{T}(B,g_T) \Bigg\vert g_T\le \frac{1}{(\bar{U}_{T-1})'(B)} \right]\Pr\left\{g_T\le \frac{1}{(\bar{U}_{T-1})'(B)} \right\} \\
        &+ \E\left[J_{T}^\text{relax}(B,g_T) - U_{T}(B,g_T) \Bigg\vert  g_T \ge \frac{e^{B}}{(\bar{U}_{T-1})'(0)}\right]\Pr\left\{g_T \ge \frac{e^{B}}{(\bar{U}_{T-1})'(0)}\right\}
        \end{split}
    \end{equation}

     When $\frac{1}{(\bar{U}_{T-1})'(B)} < g_T < \frac{e^{B}}{(\bar{U}_{T-1})'(0)}$, $b_T^\text{relax}(B,g_T)=b_T^\text{untruncated}(B,g_T)$ by \eqref{eq:bt_relax_det}, i.e., no boundary cases occur,
     \begin{equation}
        J_{T}^\text{relax}(B,g_T) - U_{T}(B,g_T)=\bar{J}_{T-1}^\text{relax}(B-b_T^\text{relax}) - \bar{U}_{T-1}(B-b_T^\text{relax}),
     \end{equation}
     where $b_T^\text{relax}=\frac{1}{T}B + \frac{T-1}{T}\log g_T\G(\nu_{T-1},\cdots, \nu_1)$ by \eqref{eq:bt_relax} for $g_T \in \left(\frac{1}{(\bar{U}_{T-1})'(B)},\; \frac{e^{B}}{(\bar{U}_{T-1})'(0)}\right)$. Thus,
     \begin{equation}
        B-b_T^\text{relax} = \frac{T-1}{T}B - \frac{T-1}{T}\log g_T \G(\nu_{T-1},\cdots, \nu_1),
     \end{equation}
     which is strictly increasing in $B$ and unbounded. Therefore, there exists $\mathfrak{B}_1$ such that $B>\mathfrak{B}_1$ implies $B-b_T^\text{relax}>\mathfrak{B}_0$ uniformly for all $g_T\in [g_{\min},g_{\max}]$. Thus, $\bar{J}_{T-1}^\text{relax}(B-b_T^\text{relax}) - \bar{U}_{T-1}(B-b_T^\text{relax})<\epsilon$ for $B>\mathfrak{B}_1$ uniformly for all $g_T\in [g_{\min},g_{\max}]$ by \eqref{eq:Jrelax_Ut_hypo} and consequently,
     \begin{equation}
     \begin{split}
        \lim_{B\to\infty}\E\Bigg[J_{T}^\text{relax}(B,g_T) - U_{T}(B,g_T) &\Bigg\vert \frac{1}{(\bar{U}_{T-1})'(B)} < g_T < \frac{e^{B}}{(\bar{U}_{T-1})'(0)}\Bigg]\Pr\left\{\frac{1}{(\bar{U}_{T-1})'(B)} < g_T < \frac{e^{B}}{(\bar{U}_{T-1})'(0)}\right\} \\
        &= \lim_{B\to\infty}\int_{\frac{1}{(\bar{U}_{T-1})'(B)}}^{\frac{e^{B}}{(\bar{U}_{T-1})'(0)}} \left[J_{T}^\text{relax}(B,x) - U_{T}(B,x)\right]  f(x) dx\\
        &= \lim_{B\to\infty}\int_{g_{\min}}^{g_{\max}} \left[J_{T}^\text{relax}(B,x) - U_{T}(B,x)\right]  f(x) dx\\
        &= \lim_{B\to\infty}\int_{g_{\min}}^{g_{\max}} \left[\bar{J}_{T-1}^\text{relax}(B-b_T^\text{relax}) - \bar{U}_{T-1}(B-b_T^\text{relax})\right]  f(x) dx=0
     \end{split}
     \end{equation}

    For sufficiently large $B$, $\frac{1}{(\bar{U}_{T-1})'(B)} < g_{\min}$ since $(\bar{U}_{T-1})'(B)=e^{B/T}\G(\nu_{T-1},\cdots, \nu_1)$ and $ \frac{e^{B}}{(\bar{U}_{T-1})'(0)}> g_{\max}$, and thus
    \begin{equation}
        \Pr\left\{g_T\le \frac{1}{(\bar{U}_{T-1})'(B)} \right\}  = \Pr\left\{g_T \ge \frac{e^{B}}{(\bar{U}_{T-1})'(0)}\right\} =0.
    \end{equation}

    Consequently, the induction follows.

    \item[(ii)] Second we show $\lim_{B\to\infty} \left[\bar{U}_T(B) - \bar{J}_T^\text{opt}(B)\right]=0$ by induction again.

        At $t=1$, all the bits are to be served and thus:
    \begin{equation}
        \bar{J}_1^\text{opt}(\beta)=\E\left[\frac{e^{\beta}-1}{g}\right]= e^{\beta}\nu_1-\nu_1 = \bar{U}_1(\beta),\qquad \forall \beta (\ge 0),
    \end{equation}
    where $\nu_1$ is defined in \eqref{eq:nu_m}.
    As an induction hypothesis, we assume that
    \begin{equation} \label{eq:Ut1_Jt1}
        \lim_{\beta\to\infty}\left[\bar{U}_{t-1}(\beta) - \bar{J}_{t-1}^\text{opt}(\beta)\right]=0.
    \end{equation}
    From \eqref{eq:Jt_opt} and \eqref{eq:bt_opt_det}, we write the expected cost-to-go as:
    \begin{equation}
        \bar{J}_{t}^\text{opt}(\beta)
        = \int_0^{\frac{1}{(\bar{J}_{t-1}^\text{opt})'(\beta)}} \bar{J}_{t-1}^\text{opt}(\beta) f(x)dx
        + \int_{\frac{e^\beta}{(\bar{J}_{t-1}^\text{opt})'(0)}}^\infty \frac{e^\beta-1}{x} f(x) dx
        + \int_{\frac{1}{(\bar{J}_{t-1}^\text{opt})'(\beta)}}^{\frac{e^\beta}{(\bar{J}_{t-1}^\text{opt})'(0)}}
        \left[\frac{e^{b_{t}^\text{opt}}-1}{x} + \bar{J}_{t-1}^\text{opt}(\beta-b_{t}^\text{opt}) \right] dF(x),
    \end{equation}
    where $b_{t}^\text{opt}$ is a function of $\beta$ (and $x$).
    Since  $(\bar{J}_{t-1}^\text{opt})'$ is unbounded increasing and $\text{Support}(f)=[g_{\min}, g_{\max}]$, $\frac{1}{(\bar{J}_{t-1}^\text{opt})'(\beta)}< g_{\min}$ and $\frac{e^\beta}{(\bar{J}_{t-1}^\text{opt})'(0)}>g_{\max}$ for sufficiently large $\beta$ as did in Lemma \ref{lem:bt_conv_Ud_Jd}, and thus
    \begin{eqnarray}
        &&\lim_{\beta\to\infty} \int_0^{\frac{1}{(\bar{J}_{t-1}^\text{opt})'(\beta)}} \bar{J}_{t-1}^\text{opt}(\beta) f(x) dx =0\\
        && \lim_{\beta\to\infty}  \int_{\frac{e^\beta}{(\bar{J}_{t-1}^\text{opt})'(0)}}^\infty \frac{e^\beta-1}{x} f(x) dx  = 0.
    \end{eqnarray}
    From Theorem \ref{thm:relax_policy_conv} and the induction hypothesis \eqref{eq:Ut1_Jt1},
    \begin{eqnarray}
        &&\lim_{\beta\to\infty} \left[ b_t^\text{relax}(\beta,g)-b_t^\text{opt}(\beta,g)\right]=0\quad \text{uniformly}\;\; \forall g\in [g_{\min}, g_{\max}],\\
        &&\lim_{\beta\to\infty} \left[\bar{U}_{t-1}(\beta-b_t^\text{relax}(\beta, g)) -\bar{J}_{t-1}^\text{opt}(\beta-b_t^\text{opt}(\beta,g))\right] =0\quad \text{uniformly}\;\; \forall g\in [g_{\min}, g_{\max}],
    \end{eqnarray}
    and thus,
    \begin{equation}
        \lim_{\beta\to\infty} \bar{J}_t^\text{opt}(\beta)= \lim_{\beta\to\infty}
        \int_{\frac{1}{(\bar{J}_{t-1}^\text{opt})'(\beta)}}^{\frac{e^\beta}{(\bar{J}_{t-1}^\text{opt})'(0)}}
        \left[\frac{e^{b_{t}^\text{relax}}-1}{x} + \bar{U}_{t-1}(\beta-b_{t}^\text{relax}) \right] dF(x)
    \end{equation}
    Therefore,
    \begin{equation}
        \begin{split}
            \lim_{\beta\to\infty} [\bar{U}_t(\beta) &- \bar{J}_t^\text{opt}(\beta)]
            = \lim_{\beta\to\infty} \left[t e^{\frac{\beta}{t}}\G\left(\nu_t,\nu_{t-1},\cdots, \nu_1\right)-t\nu_1 - \int_{\frac{1}{(\bar{J}_{t-1}^\text{opt})'(\beta)}}^{\frac{e^\beta}{(\bar{J}_{t-1}^\text{opt})'(0)}}
        \left[\frac{e^{b_{t}^\text{relax}}-1}{x} + \bar{U}_{t-1}(\beta-b_{t}^\text{relax}) \right] dF(x) \right] \\
            &=\lim_{\beta\to\infty} \left[t e^{\frac{\beta}{t}}\G\left(\nu_t,\nu_{t-1},\cdots, \nu_1\right)-\int_{\frac{1}{(\bar{J}_{t-1}^\text{opt})'(\beta)}}^{\frac{e^\beta}{(\bar{J}_{t-1}^\text{opt})'(0)}}
        \left[\frac{e^{b_{t}^\text{relax}}-1}{x} + (t-1) e^{\frac{\beta-b_t^\text{relax}}{t-1}} \G\left(\nu_{t-1},\cdots,\nu_1\right) \right] dF(x) \right]
        \end{split}
    \end{equation}
    By substituting \eqref{eq:bt_relax} into $b_t^\text{relax}$, we have $\lim_{\beta\to\infty} [\bar{U}_t(\beta) - \bar{J}_t^\text{opt}(\beta)]=0$ as desired. Thus, the induction holds.
\end{enumerate}
    By (i) and (ii), we obtain the result as desired.

\section{Proof of Theorem \ref{thm:bt_one_conv_policy}}
\label{sec:pf_bt_one_conv_policy}
    First, we show that the optimal scheduler becomes a threshold policy as $\beta\to 0$. As in \eqref{eq:bt_opt_det}, the optimal policy is determined differently by the range of $g_t$: $g_t\le 1/(\bar{J}^\text{opt}_{t-1})'(\beta)$, $1/(\bar{J}^\text{opt}_{t-1})'(\beta) < g_t < e^{\beta}/(\bar{J}^\text{opt}_{t-1})'(0)$, or $g_t\ge e^{\beta}/(\bar{J}^\text{opt}_{t-1})'(0)$. Since $\lim_{\beta\to 0} e^\beta=1$ and $\lim_{\beta\to 0} (\bar{J}^\text{opt}_{t-1})'(\beta)= (\bar{J}^\text{opt}_{t-1})'(0)$,
    \begin{equation} \label{eq:lim_bounds}
        \lim_{\beta\to 0} \left[\frac{e^\beta}{(\bar{J}_{t-1}^\text{opt})'(0)}-\frac{1}{(\bar{J}_{t-1}^\text{opt})'(\beta)}\right]=0,
    \end{equation}
    which implies that the case of $1/(\bar{J}^\text{opt}_{t-1})'(\beta) < g_t < e^{\beta}/(\bar{J}^\text{opt}_{t-1})'(0)$ occurs with vanishing probability as $\beta\to 0$. Thus, the optimal policy is a threshold policy , i.e.,
    \begin{equation} \label{eq:bt_opt_threshold}
        b_t^\text{opt}(\beta,g_t) = \begin{cases}
            \beta, & g_t> \frac{1}{(\bar{J}_{t-1}^\text{opt})'(0)}, \\
            0, & g_t\le \frac{1}{(\bar{J}_{t-1}^\text{opt})'(0)}
        \end{cases}
\end{equation}
    as $\beta\to 0$. This implies that
    \begin{equation}
        \lim_{\beta\to 0} \sup \{g: b_t^\text{opt}(\beta,g) =0\} = \lim_{\beta\to 0} \inf \{g : b_t^\text{opt}(\beta,g)=\beta \} = \frac{1}{(\bar{J}_{t-1}^\text{opt})'(0)}.
    \end{equation}

    Second, we show that the thresholds are identical, i.e., $\lim_{\beta\to 0} (\bar{J}_{t-1}^\text{opt})'(\beta)=\omega_t$ for every $t$, where $\omega_t$ is defined in \eqref{eq:stopping_omega_t}. When $t=2$, this holds by construction.
    As an induction hypothesis, we suppose that $\lim_{\beta\to 0} (\bar{J}_{t-2}^\text{opt})'(\beta)=(\bar{J}_{t-2}^\text{opt})'(0)=\omega_{t-1}$. By \eqref{eq:Jt1d} and \eqref{eq:Jt1d_sub},
    \begin{equation} 
        \begin{split}
            \lim_{\beta\to 0} (\bar{J}_{t-1}^\text{opt})'(\beta)
            &=\lim_{\beta\to 0}\left[\int_0^{\frac{1}{(\bar{J}_{t-2}^\text{opt})'(\beta)}} (\bar{J}_{t-2}^\text{opt})'(\beta) dF(x)
        + e^\beta \int_{\frac{e^\beta}{(\bar{J}_{t-2}^\text{opt})'(0)}}^\infty \frac{1}{x} dF(x)
        + \int_{\frac{1}{(\bar{J}_{t-2}^\text{opt})'(\beta)}}^{\frac{e^\beta}{(\bar{J}_{t-2}^\text{opt})'(0)}}
        \frac{e^{b_{t-1}^\text{opt}}}{x} dF(x)\right] \\
            &= (\bar{J}_{t-2}^\text{opt})'(0) \int_0^{\frac{1}{(\bar{J}_{t-2}^\text{opt})'(0)}}
            dF(x) + \int_{\frac{1}{(\bar{J}_{t-2}^\text{opt})'(0)}}^\infty \frac{1}{x} dF(x) \\
            &= (\bar{J}_{t-2}^\text{opt})'(0) \Pr\left\{g_t \le \frac{1}{(\bar{J}_{t-2}^\text{opt})'(0)}\right\} +
             \E\left[\frac{1}{g_t} \Bigg\vert g_t > \frac{1}{(\bar{J}_{t-2}^\text{opt})'(0)}\right] \Pr\left\{g_t > \frac{1}{(\bar{J}_{t-2}^\text{opt})'(0)}\right\}= \omega_{t}
        \end{split}
    \end{equation}
    where the last equality follows from \eqref{eq:stopping_omega_t} by substituting $\omega_{t-1}$ into $(\bar{J}_{t-2}^\text{opt})'(0)$ from the induction hypothesis.
    Thus, the induction holds.

\section{Proof of Theorem \ref{thm:oneshot_cost_conv}}
\label{sec:pf_oneshot_cost_conv}
    Since $\bar{J}_1^\text{one}(B)=(e^B-1)\E\left[\frac{1}{g}\right]$, $(\bar{J}_1^\text{one})'(0)=\E\left[\frac{1}{g}\right]=\omega_2$ by \eqref{eq:stopping_omega_t}. If we suppose that $\lim_{B\to 0} (\bar{J}_{t-2}^\text{one})'(B)=(\bar{J}_{t-2}^\text{one})'(0)=\omega_{t-1}$, then from \eqref{eq:Jt_one_stopping}
    \begin{eqnarray}
        \bar{J}_{t-1}^\text{one}(B) &=& \int_0^{\frac{1}{\omega_{t-1}}} \bar{J}_{t-2}^\text{one}(B) dF(x) + \int_{\frac{1}{\omega_{t-1}}}^\infty \frac{e^B-1}{x}dF(x) \\
        (\bar{J}_{t-1}^\text{one})'(B) &=& \int_0^{\frac{1}{\omega_{t-1}}} (\bar{J}_{t-2}^\text{one})'(B) dF(x) + \int_{\frac{1}{\omega_{t-1}}}^\infty \frac{e^B}{x}dF(x).
    \end{eqnarray}
    Thus,
    \begin{equation}
        \begin{split}
            \lim_{B\to 0} (\bar{J}_{t-1}^\text{one})'(B)
            &= (\bar{J}_{t-2}^\text{one})'(0) \int_0^{\frac{1}{\omega_{t-1}}} dF(x) + \int_{\frac{1}{\omega_{t-1}}}^\infty \frac{1}{x} dF(x) \\
            &= \omega_{t-1}\Pr\left\{\frac{1}{g}\ge \omega_{t-1}\right\} + \E\left[\frac{1}{g}\Bigg\vert\frac{1}{g}\ge \omega_{t-1} \right]\Pr\left\{\frac{1}{g}\ge \omega_{t-1}\right\} = \omega_t.
        \end{split}
    \end{equation}
    By induction, $\lim_{B\to 0} (\bar{J}_{T}^\text{one})'(B)=\omega_{T+1}$. In the proof of Theorem \ref{thm:bt_one_conv_policy}, we have shown that $\lim_{B\to 0} (\bar{J}_{T}^\text{opt})'(B)=\omega_{T+1}$ also. Since $\lim_{B\to 0} \bar{J}_T^\text{one}(B) = \lim_{B\to 0} \bar{J}_T^\text{opt}(B)=0$, by L'Hopital's rule, we have
    \begin{equation}
        \lim_{B\to 0} \frac{\bar{J}_T^\text{one}(B)}{\bar{J}_T^\text{opt}(B)} =\lim_{B\to 0} \frac{(\bar{J}_T^\text{one})'(B)}{(\bar{J}_T^\text{opt})'(B)}=1.
    \end{equation}

\section{Proof of Theorem \ref{thm:J_ergdeltaopt_conv}}
\label{sec:pf_J_ergdeltaopt_conv}
    By definition,
    \begin{equation} \label{eq:inequalities_Jerg_Jergdeltaopt}
        \bar{E}^\text{erg}(\bar{b}) \le \frac{1}{T}\bar{J}_T^\text{opt}(\bar{b}T)\le 
        \frac{1}{T}\bar{J}_T^\text{constrained-erg}(\bar{b}T).
    \end{equation}
    Let $\epsilon>0$ be given. Since $\bar{E}^\text{erg}(\bar{b})$ is an increasing continuous function of $\bar{b}$, there exists $\delta>0$ such that
    \begin{equation}
        \bar{E}^\text{erg}(\bar{b})+\epsilon = \bar{E}^\text{erg}(\bar{b}+\delta).
    \end{equation}
    We use this $\delta$ for $b_t^\text{erg-delta}$. Then, by \eqref{eq:barJT_constrained_delta} and \eqref{eq:barJT_constrained},
    \begin{equation}
        \begin{split}
        \frac{1}{T} \bar{J}_T^\text{constrained-erg}(\bar{b}T)
                &\le \frac{1}{T}\E\left[ \sum_{t=2}^T \frac{e^{b_t^\text{erg}(\bar{b}+\delta, g_t)}-1}{g_t} + \frac{e^{\beta_1}-1}{g_1} \right] \\
                &= \frac{T-1}{T}\E\left[\frac{1}{T-1} \sum_{t=2}^T \frac{e^{b_t^\text{erg}(\bar{b}+\delta, g_t)}-1}{g_t} \right]+ \frac{1}{T}\E_{g_1}\left[\E_{\beta_1}\left[\frac{e^{\beta_1}-1}{g_1} \right]\right]
        \end{split}
    \end{equation}
    Notice that $\left\{\frac{e^{b_t^\text{erg}(\bar{b}+\delta, g_t)}-1}{g_t}\right\}_{t=2}^T$ are i.i.d.~and thus,
    \begin{equation}
        \E\left[\frac{1}{T-1} \sum_{t=2}^T \frac{e^{b_t^\text{erg}(\bar{b}+\delta, g_t)}-1}{g_t}\right] = \E \left[\frac{e^{b_t^\text{erg}(\bar{b}+\delta, g_t)}-1}{g_t}\right] = \bar{E}^\text{erg}(\bar{b}+\delta).
    \end{equation}
    Since  $\{b_t^\text{erg}\}_{t=2}^T$ are i.i.d., $\frac{1}{T-1}\sum_{t=2}^T b_t^\text{erg} \to \E[b_t^\text{erg}]=\bar{b}+\delta$ almost surely (a.s.) as $T\to \infty$ by the law of large number, and thus, the remaining bits at the final slot is given by
    \begin{equation}
        \bar{b}T-\sum_{t=2}^T b_t^\text{erg}=(T-1)\left(\frac{T}{T-1}\bar{b}-\frac{1}{T-1}\sum_{t=2}^T b_t^\text{erg}\right) \le \bar{b}+\delta\quad \text{a.s.}
    \end{equation}
    That is, $e^{\beta_1}\le e^{\bar{b}+\delta}$ a.s.~and therefore $\E\left[e^{\beta_1}\right]\le e^{\bar{b}+\delta}$.
    \begin{equation}
        \E_{g_1}\left[\E_{\beta_1}\left[\frac{e^{\beta_1}-1}{g_1} \right]\right]
        \le \E_{g_1}\left[\frac{e^{\bar{b}+\delta}-1}{g_1}\right]
        = \min_{b=\bar{b}+\delta}\;\;\E_{g_1}\left[\frac{e^{b}-1}{g_1}\right]
        \le \min_{\E[b]=\bar{b}+\delta}\;\;\E\left[\frac{e^{b}-1}{g_1}\right] = \bar{E}^\text{erg}(\bar{b}+\delta)
    \end{equation}
    Thus,
    \begin{equation}
        \lim_{T\to\infty} \frac{1}{T} J_T^\text{constrained-erg}(\bar{b}T)
        \le \lim_{T\to\infty}\left[\frac{T-1}{T} \bar{E}^\text{erg}(\bar{b}+\delta) + \frac{1}{T} \bar{E}^\text{erg}(\bar{b}+\delta)\right] = \bar{E}^\text{erg}(\bar{b}+\delta).
    \end{equation}
    Therefore,
    \begin{equation}
        \bar{E}^\text{erg}(\bar{b}) \le \lim_{T\to\infty}\frac{1}{T}\bar{J}_T^\text{constrained-erg}(\bar{b}T)\le \bar{E}^\text{erg}(\bar{b})+\epsilon.
    \end{equation}
    Since $\epsilon$ is arbitrary, we have the result.

\section{Proof of Theorem \ref{thm:scheduling_gain_bounds}}
\label{sec:pf_scheduling_gain_bounds}
    First, we show the monotonicity of $\Delta_T^\text{opt}(B)$. Since
    \begin{equation} \label{eq:dDelta_T}
    \frac{d}{dB}\Delta_T^\text{opt}(B) = \frac{(\bar{J}_T^\text{eq})'(B)\bar{J}_T^\text{opt}(B)-\bar{J}_T^\text{eq}(B)(\bar{J}_T^\text{opt})'(B)}{\left(\bar{J}_T^\text{opt}(B)\right)^2},
    \end{equation}
    we will investigate the quantity $(\bar{J}_T^\text{eq})'(B)\bar{J}_T^\text{opt}(B)-\bar{J}_T^\text{eq}(B)(\bar{J}_T^\text{opt})'(B)$.

    From \eqref{eq:Jt_opt}, \eqref{eq:bt_opt_det}, and \eqref{eq:JT_eq}, we have
\begin{eqnarray}
    \bar{J}_T^\text{opt}(B) &=& \int_0^{\frac{1}{(\bar{J}_{T-1}^\text{opt})'(B)}} \bar{J}_{T-1}^\text{opt}(B) dF(x)
     + \int_{\frac{1}{(\bar{J}_{T-1}^\text{opt})'(B)}}^{\frac{e^B}{(\bar{J}_{T-1}^\text{opt})'(0)}} \left[\frac{e^{b_T^\text{opt}}-1}{x}+\bar{J}_{T-1}^\text{opt}(B-b_{T}^\text{opt}) \right] dF(x)\nonumber \\
    &&+ \int_{\frac{e^B}{(\bar{J}_{T-1}^\text{opt})'(0)}}^\infty \frac{e^B-1}{x} dF(x) \\
    (\bar{J}_T^\text{opt})'(B) &=& \int_0^{\frac{1}{(\bar{J}_{T-1}^\text{opt})'(B)}} (\bar{J}_{T-1}^\text{opt})'(B) dF(x)
    + \int_{\frac{1}{(\bar{J}_{T-1}^\text{opt})'(B)}}^{\frac{e^B}{(\bar{J}_{T-1}^\text{opt})'(0)}} \frac{d}{dB}\left[\frac{e^{b_T^\text{opt}}-1}{x}+\bar{J}_{T-1}^\text{opt}(B-b_{T}^\text{opt}) \right] dF(x) \nonumber \\
    &&+ \int_{\frac{e^B}{(\bar{J}_{T-1}^\text{opt})'(0)}}^\infty \frac{e^B}{x} dF(x) \\
    \bar{J}_T^\text{eq}(B) &=& T(e^{\frac{B}{T}}-1)\nu_1 \label{eq:barJT_eq_B}\\
    (\bar{J}_T^\text{eq})'(B) &=& e^{\frac{B}{T}}\nu_1 \label{eq:dbardJT_eq_B}
\end{eqnarray}
    When $T=2$,
    \begin{equation}
    \begin{split}
    (\bar{J}_2^\text{eq})'(B)\bar{J}_2^\text{opt}(B) - \bar{J}_2^\text{eq}(B)(\bar{J}_2^\text{opt})'(B)
    &= \int_0^{\frac{e^{-B}}{\nu_1}} e^{\frac{B}{2}}\nu_1^2 (-1)(e^{\frac{B}{2}}-1)^2 f(x) dx  \\
    &+\int_{\frac{e^{-B}}{\nu_1}}^{\frac{e^{B}}{\nu_1}} e^{\frac{B}{2}}\nu_1 (-1)\left[\left(\frac{1}{x}\right)^{\frac{1}{2}}-\left(\frac{1}{\nu_1}\right)^{\frac{1}{2}}\right]^2 f(x) dx \\
    &+ \int_{\frac{e^{B}}{\nu_1}}^\infty \frac{e^{\frac{B}{2}}\nu_1}{x} (-1)(e^{\frac{B}{2}}-1)^2 f(x) dx \\
    &\le 0.
    \end{split}
    \end{equation}
    That is, $\frac{d\Delta_2^\text{opt}(B)}{dB}\le 0$.

We now suppose that $\frac{d\Delta_{T-1}^\text{opt}(B)}{dB}\le 0$ and examine $\frac{d\Delta_T^\text{opt}(B)}{dB}$. That is, we assume that
    \begin{equation} \label{eq:num_dDelta_T1}
        (\bar{J}_{T-1}^\text{eq})'(B)\bar{J}_{T-1}^\text{opt}(B) - \bar{J}_{T-1}^\text{eq}(B)(\bar{J}_{T-1}^\text{opt})'(B) \le 0,
    \end{equation}
    where the left hand side is the numerator of $\frac{d\Delta_{T-1}^\text{opt}(B)}{dB}$ from \eqref{eq:dDelta_T}. The numerator of $\frac{d\Delta_{T}^\text{opt}(B)}{dB}$ is given by
    \begin{equation} \label{eq:num_dDelta}
    	\footnotesize
		\begin{split}
    	(\bar{J}_T^\text{eq})'(B)\bar{J}_T^\text{opt}(B) &- \bar{J}_T^\text{eq}(B)(\bar{J}_T^\text{opt})'(B)\\
    &= \int_0^{\frac{1}{(\bar{J}_{T-1}^\text{opt})'(B)}} \nu_1 \left[e^{\frac{B}{T}} \bar{J}_{T-1}^\text{opt}(B)- T(e^{\frac{B}{T}}-1) (\bar{J}_{T-1}^\text{opt})'(B) \right] f(x) dx \\
    &+ \int_{\frac{e^B}{(\bar{J}_{T-1}^\text{opt})'(0)}}^\infty \nu_1 \left[e^{\frac{B}{T}} \frac{e^B-1}{x}  - T(e^{\frac{B}{T}}-1) \frac{e^B}{x}  \right] f(x) dx \\
    &+ \int_{\frac{1}{(\bar{J}_{T-1}^\text{opt})'(B)}}^{\frac{e^B}{(\bar{J}_{T-1}^\text{opt})'(0)}}  \nu_1\left[e^{\frac{B}{T}} \left(\frac{e^{b_T^\text{opt}}-1}{x}+\bar{J}_{T-1}^\text{opt}(B-b_{T-1}^\text{opt}) \right) - T(e^{\frac{B}{T}}-1) \frac{d}{dB}\left(\frac{e^{b_T^\text{opt}}-1}{x}+\bar{J}_{T-1}^\text{opt}(B-b_{T-1}^\text{opt}) \right)  \right] f(x) dx 
		\end{split}
    \end{equation}
	From the integrand of the first integral in \eqref{eq:num_dDelta},
    \begin{multline} \label{eq:first_int}
        e^{\frac{B}{T}}\bar{J}_{T-1}^\text{opt}(B) -T(e^{\frac{B}{T}}-1)(\bar{J}_{T-1}^\text{opt})'(B)
            =  e^{-\frac{B}{T(T-1)}}\left[e^{\frac{B}{T-1}}\bar{J}_{T-1}^\text{opt}(B) - (T-1)(e^{\frac{B}{T-1}}-1)(\bar{J}_{T-1}^\text{opt})'(B)\right]\\+\left[T-e^{\frac{B}{T}} - (T-1)e^{-\frac{B}{T(T-1)}} \right](\bar{J}_{T-1}^\text{opt})'(B)
    \end{multline}
    From \eqref{eq:barJT_eq_B}, \eqref{eq:dbardJT_eq_B}, and the hypothesis \eqref{eq:num_dDelta_T1}, we have
    \begin{equation} \label{eq:first_int_1}
    	e^{\frac{B}{T-1}}\bar{J}_{T-1}^\text{opt}(B) - (T-1)(e^{\frac{B}{T-1}}-1)(\bar{J}_{T-1}^\text{opt})'(B) \le 0.
    \end{equation}
    We define a concave function $\phi$ such that
    \begin{equation}
        \phi(z)=T-z^{T-1}-(T-1)z^{-1},\qquad z > 0.
    \end{equation}
    Since the concavity and $\phi'(z)=0$ yield that $\phi$ attains its maximum at $z=1$ and $\phi(1)=0$, $\phi(z)\le 0$ for all $z>0$. Since $\phi(e^{\frac{B}{T(T-1)}})=T-e^{\frac{B}{T}} - (T-1)e^{-\frac{B}{T(T-1)}}$,
    \begin{equation} \label{eq:first_int_2}
    	T-e^{\frac{B}{T}} - (T-1)e^{-\frac{B}{T(T-1)}} \le 0.
    \end{equation}
    From \eqref{eq:first_int_1} and \eqref{eq:first_int_2} along with the fact that $\bar{J}_{T-1}^\text{opt}$ is convex, \eqref{eq:first_int} becomes
    \begin{equation}
        e^{\frac{B}{T}}\bar{J}_{T-1}^\text{opt}(B) -T(e^{\frac{B}{T}}-1)(\bar{J}_{T-1}^\text{opt})'(B)\le 0.
    \end{equation}
	
	Likewise, from the integrand of the second integral in \eqref{eq:num_dDelta}, we want to show that 
	\begin{equation} \label{eq:int_2}
		e^{\frac{B}{T}} (e^B-1)  - T(e^{\frac{B}{T}}-1) e^B \le 0,
	\end{equation}
	which is equivalent to show $e^B-1-T(1-e^{-\frac{B}{T}})e^B\le 0$. If we define a concave function $\psi(z)=-(T-1)z^T+T z^{T-1}-1$ for $z>0$, $\psi(e^{\frac{B}{T}})=e^B-1-T(1-e^{-\frac{B}{T}})e^B$. As we did before, we can show that $\psi(z)\le 0$, and thus \eqref{eq:int_2} holds. 
	
	From the integrand of the third integral in \eqref{eq:num_dDelta}, we want to show that
	\begin{eqnarray}
		e^{\frac{B}{T}}\left(e^{b_T^\text{opt}}-1\right) - T\left(e^{\frac{B}{T}}-1\right) e^{b_T^\text{opt}}\le 0, \label{eq:int_3_part}\\
		e^{\frac{B}{T}} \bar{J}_{T-1}^\text{opt}(B-b_{T-1}^\text{opt}) - T\left(e^{\frac{B}{T}}-1\right) \frac{d}{dB} \bar{J}_{T-1}^\text{opt}(B-b_{T-1}^\text{opt})\le 0. \label{eq:int_3_part_2}
	\end{eqnarray}
	To prove \eqref{eq:int_3_part}, we can write
	\begin{multline}
		e^{\frac{B}{T}}\left(e^{b_T^\text{opt}}-1\right) - T\left(e^{\frac{B}{T}}-1\right) e^{b_T^\text{opt}}
		= \\ e^{\frac{B-b_T^\text{opt}}{T}}\left[e^{\frac{b_T^\text{opt}}{T}}\left(e^{b_T^\text{opt}}-1\right)-T\left(e^{\frac{b_T^\text{opt}}{T}}-1\right)e^{b_T^\text{opt}}\right] - e^{\frac{B-b_T^\text{opt}}{T}}Te^{b_T^\text{opt}}\left(1-e^{-\frac{B-b_T^\text{opt}}{T}}\right)
	\end{multline}
Notice that \eqref{eq:int_2} holds for every $B\ge 0$ and thus 
\begin{equation}
	e^{\frac{b_T^\text{opt}}{T}} (e^{b_T^\text{opt}}-1)  - T(e^{\frac{b_T^\text{opt}}{T}}-1) e^{b_T^\text{opt}} \le 0.
\end{equation}
Therefore, \eqref{eq:int_3_part} holds. By the similar argument, \eqref{eq:int_3_part_2} holds, too. Thus, the third integral in \eqref{eq:num_dDelta} is no greater than 0. Consequently, we obtain
\begin{equation}
	(\bar{J}_T^\text{eq})'(B)\bar{J}_T^\text{opt}(B) - \bar{J}_T^\text{eq}(B)(\bar{J}_T^\text{opt})'(B) \le 0.
\end{equation}
This shows that the monotonicity of $\Delta_T^\text{opt}(B)$ inductively.

    Second, the limits are calculated with $\bar{J}_T^\text{relax}(B)$ and $\bar{J}_T^\text{one}(B)$ since $\bar{J}_T^\text{opt}(B)$ converges to the former for large $B$ and to the latter for small $B$ by Theorem \ref{thm:relax_cost_conv} and Theorem \ref{thm:oneshot_cost_conv}, respectively.

\section{Derivation of High SNR Affine Approximation Parameters}
\label{sec:deriv_high_snr_approx}
From \eqref{eq:JT_eq}, the spectral efficiency of the equal-bit scheduler can be found by
\begin{equation} \label{eq:RT_eq}
    P=(e^{R_T^\text{eq}}-1)\nu_1\quad \text{or} \quad R_T^\text{eq}(P) = \log \left(1+\frac{P}{\nu_1}\right).
\end{equation}
Thus, at high SNR
\begin{equation}
    R_T^\text{eq}(P)=\log P -\log \nu_1 + o(1).
\end{equation}
Similarly, from \eqref{eq:barUt} and Theorem \ref{thm:relax_cost_conv}, the spectral efficiency of the optimal scheduler at high SNR is given by
\begin{equation}
    R_T^\text{opt}(P) = \log P - \log \G\left(\nu_T, \cdots, \nu_1\right) + o(1).
\end{equation}
At high SNR, the ergodic capacity can be approximately given by the uniform power control:
\begin{equation}
    \begin{split}
    R^\text{erg}(P) &\approx  \E\left[\log (1+gP)\right] \\
    &\approx \log P + \E\left[\log g\right] \\
    &= \log P - \log e^{\E\left[\log \left(\frac{1}{g}\right)\right]} \\
    &=\log P - \log \nu_\infty,
    \end{split}
\end{equation}
where the last equality follows from
\begin{equation}
    e^{\E\left[\log \left(\frac{1}{g}\right)\right]} = \lim_{x\to 0} e^{\frac{1}{x}\log \E\left[\left(\frac{1}{g}\right)^x\right]} = \lim_{m\to\infty} \left(\E\left[\left(\frac{1}{g}\right)^{\frac{1}{m}}\right]\right)^m = \nu_\infty.
\end{equation}

\section{Derivation of Low SNR Affine Approximation Parameters}
\label{sec:deriv_low_snr_approx}
From \eqref{eq:RT_eq},
\begin{equation}
    \dot{R}^\text{eq}(P) = \frac{1/\nu_1}{1+\frac{P}{\nu_1}}.
\end{equation}
By \cite{Verdu_IT02},
\begin{equation}
    \left(\frac{E_b}{N_0}\right)_{\min}^\text{eq} = \frac{\log 2}{\dot{R}^\text{eq}(0)}= (\log 2) \nu_1.
\end{equation}
Then second order analysis is given by
\begin{equation}
    \ddot{R}^\text{eq}(P) = \frac{-\left(1/\nu_1\right)^2}{\left(1+\frac{P}{
    \nu_1}\right)^2}
\end{equation}
Therefore, by \cite{Verdu_IT02},
\begin{equation}
    \mathcal{S}_0^\text{eq}= - 2\frac{\left(\dot{R}^\text{eq}(0)\right)^2}{\ddot{R}^\text{eq}(0)}=2.
\end{equation}
By Theorem \ref{thm:oneshot_cost_conv}, the average total energy cost of the optimal scheduler at low SNR is given by
\begin{equation}
    \bar{J}_T^\text{opt} =(e^B-1) \E\left[\min\left(\frac{1}{g_T}, \E\left[\min\left(\frac{1}{g_{T-1}}, \cdots \E\left[\min\left(\frac{1}{g_2},\E\left[\frac{1}{g_1}\right]\right)\right] \right)\right]\right)\right].
\end{equation}
With the per slot basis notations,
\begin{equation}
    TP = (e^{TR^\text{opt}}-1) \E\left[\min\left(\frac{1}{g_T}, \E\left[\min\left(\frac{1}{g_{T-1}}, \cdots \E\left[\min\left(\frac{1}{g_2},\E\left[\frac{1}{g_1}\right]\right)\right] \right)\right]\right)\right]
\end{equation}
and thus,
\begin{equation}
    R^\text{opt}(P)= \frac{1}{T}\log \left(1+\frac{TP}{ \E\left[\min\left(\frac{1}{g_T}, \E\left[\min\left(\frac{1}{g_{T-1}}, \cdots \E\left[\min\left(\frac{1}{g_2},\E\left[\frac{1}{g_1}\right]\right)\right] \right)\right]\right)\right]}\right).
\end{equation}
Therefore, we have
\begin{eqnarray}
    \dot{R}^\text{opt}(P) &=& \frac{\frac{1}{ \E\left[\min\left(\frac{1}{g_T}, \E\left[\min\left(\frac{1}{g_{T-1}}, \cdots \E\left[\min\left(\frac{1}{g_2},\E\left[\frac{1}{g_1}\right]\right)\right] \right)\right]\right)\right]}}{1+\frac{TP}{ \E\left[\min\left(\frac{1}{g_T}, \E\left[\min\left(\frac{1}{g_{T-1}}, \cdots \E\left[\min\left(\frac{1}{g_2},\E\left[\frac{1}{g_1}\right]\right)\right] \right)\right]\right)\right]}}\\
    \ddot{R}^\text{opt}(P) &=& \frac{\frac{T}{\left( \E\left[\min\left(\frac{1}{g_T}, \E\left[\min\left(\frac{1}{g_{T-1}}, \cdots \E\left[\min\left(\frac{1}{g_2},\E\left[\frac{1}{g_1}\right]\right)\right] \right)\right]\right)\right]\right)^2}}{\left(1+\frac{TP}{ \E\left[\min\left(\frac{1}{g_T}, \E\left[\min\left(\frac{1}{g_{T-1}}, \cdots \E\left[\min\left(\frac{1}{g_2},\E\left[\frac{1}{g_1}\right]\right)\right] \right)\right]\right)\right]}\right)^2}.
\end{eqnarray}
Thus,
\begin{eqnarray}
    \left(\frac{E_b}{N_0}\right)_{\min}^\text{one} &=& \frac{\log 2}{\dot{R}(0)} = (\log 2)  \E\left[\min\left(\frac{1}{g_T}, \E\left[\min\left(\frac{1}{g_{T-1}}, \cdots \E\left[\min\left(\frac{1}{g_2},\E\left[\frac{1}{g_1}\right]\right)\right] \right)\right]\right)\right] \\
    \mathcal{S}_0^\text{one} &=& -2 \frac{\left(\dot{R}^\text{one}(0)\right)^2}{\ddot{R}^\text{one}(0)}= \frac{2}{T}.
\end{eqnarray}
See \cite{Verdu_IT02} for $\left(\frac{E_b}{N_0}\right)_{\min}$ and $\mathcal{S}_0$ of the ergodic capacity.

\bibliographystyle{IEEEtran}
\bibliography{scheduling}

\begin{thebibliography}{10}
\providecommand{\url}[1]{#1}
\csname url@rmstyle\endcsname
\providecommand{\newblock}{\relax}
\providecommand{\bibinfo}[2]{#2}
\providecommand\BIBentrySTDinterwordspacing{\spaceskip=0pt\relax}
\providecommand\BIBentryALTinterwordstretchfactor{4}
\providecommand\BIBentryALTinterwordspacing{\spaceskip=\fontdimen2\font plus
\BIBentryALTinterwordstretchfactor\fontdimen3\font minus
  \fontdimen4\font\relax}
\providecommand\BIBforeignlanguage[2]{{%
\expandafter\ifx\csname l@#1\endcsname\relax
\typeout{** WARNING: IEEEtran.bst: No hyphenation pattern has been}%
\typeout{** loaded for the language `#1'. Using the pattern for}%
\typeout{** the default language instead.}%
\else
\language=\csname l@#1\endcsname
\fi
#2}}

\bibitem{Fu_WC06}
A.~Fu, E.~Modiano, and J.~N. Tsitsiklis, ``Optimal transmission scheduling over
  a fading channel with energy and deadline constraints,'' \emph{IEEE Trans.
  Wireless Commun.}, vol.~5, no.~3, pp. 630--641, Mar. 2006.

\bibitem{Lee_WC09}
J.~Lee and N.~Jindal, ``Energy-efficient scheduling of delay constrained
  traffic over fading channels,'' \emph{IEEE Trans. Wireless Commun.}, vol.~8,
  no.~4, pp. 1866--1875, Apr. 2009.

\bibitem{Shamai_IT01}
S.~Shamai and S.~Verd\'u, ``The impact of frequency-flat fading on the spectral
  efficiency of {CDMA},'' \emph{IEEE Trans. Inform. Theory}, vol.~47, no.~5,
  May 2001.

\bibitem{Verdu_IT02}
S.~Verd\'u, ``Spectral efficiency in the wideband regime,'' \emph{IEEE Trans.
  Inform. Theory}, vol.~48, no.~6, pp. 1319--1343, Jun. 2002.

\bibitem{Zafer_WITA07}
M.~Zafer and E.~Modiano, ``Delay constrained energy efficient data transmission
  over a wireless fading channel,'' in \emph{Workshop on Inf.~Theory and
  Appl.}, La Jolla, CA, Jan./Feb. 2007, pp. 289--298.

\bibitem{Negi_IT02}
R.~Negi and J.~M. Cioffi, ``Delay-constrained capacity with causal feedback,''
  \emph{IEEE Trans. Inform. Theory}, vol.~48, no.~9, pp. 2478--2494, Sep. 2002.

\bibitem{Caire_IT04}
G.~Caire, D.~Tuninetti, and S.~Verd\'u, ``Variable-rate coding for slowly
  fading gaussian multiple-access channels,'' \emph{IEEE Trans. Inform.
  Theory}, vol.~50, no.~10, pp. 2271--2292, Oct. 2004.

\bibitem{Goldsmith_IT97}
A.~Goldsmith and P.~Varaiya, ``Capacity of fading channels with channel side
  information,'' \emph{IEEE Trans. Inform. Theory}, vol.~43, pp. 1986--1992,
  Nov. 1997.

\bibitem{Caire_IT99}
G.~Caire, G.~Taricco, and E.~Biglieri, ``Optimum power control over fading
  channels,'' \emph{IEEE Trans. Inform. Theory}, vol.~45, no.~5, pp.
  1468--1489, Jul. 1999.

\bibitem{Rockafellar_Book70}
R.~T. Rockafellar, \emph{Convex Analysis}.\hskip 1em plus 0.5em minus
  0.4em\relax Princeton Univ. Press, 1970.

\bibitem{Bertsekas_DP1_Book05}
D.~P. Bertsekas, \emph{Dynamic Programming and Optimal Control}, 3rd~ed.\hskip
  1em plus 0.5em minus 0.4em\relax Mass.: Athena Scientific, 2005, vol.~1.

\bibitem{Goldsmith_Book05}
A.~Goldsmith, \emph{Wireless Communications}.\hskip 1em plus 0.5em minus
  0.4em\relax New York, NY: Cambridge Univ. Press, 2005.

\bibitem{Bertsekas_AC75}
D.~P. Bertsekas, ``Convergence of discretization procedures in dynamic
  programming,'' \emph{IEEE Trans. Automat. Contr.}, vol. AC-20, no.~3, pp.
  415--419, Jun. 1975.

\bibitem{Hanly_IT98}
V.~Hanly and D.~Tse, ``Multiaccess fading channels. {P}art {II}: Delay-limited
  capacities,'' \emph{IEEE Trans. Inform. Theory}, vol.~44, pp. 2816--2831,
  Nov. 1998.

\bibitem{Berry_IT02}
R.~A. Berry and R.~G. Gallager, ``Communication over fading channels with delay
  constraints,'' \emph{IEEE Trans. Inform. Theory}, vol.~48, no.~5, pp.
  1135--1149, May. 2002.

\bibitem{Rajan_IT04}
D.~Rajan, A.~Sabharwal, and B.~Aazhang, ``Delay-bounded packet scheduling of
  bursty traffic over wireless channels,'' \emph{IEEE Trans. Inform. Theory},
  vol.~50, no.~1, pp. 125--144, Jan. 2004.

\bibitem{Yeh_ISIT03}
E.~M. Yeh and A.~S. Cohen, ``Throughput and delay optimal resource allocation
  in multiaccess fading channels,'' in \emph{Proc. IEEE Int. Symp. on Inform.
  Theory (ISIT)}, Yokohama, Japan, Jun./Jul. 2003, p. 245.

\end{thebibliography}
\end{document}